\documentclass[journal]{IEEEtran}
\usepackage{cite}
\usepackage{amsmath,amssymb,amsfonts}
\usepackage{algorithmicx}
\usepackage[noend,ruled, linesnumbered, vlined]{algorithm2e}
\usepackage{pifont}
\usepackage{array}
\usepackage{graphicx}
\usepackage{stfloats}
\usepackage{amsthm}
\usepackage{subfigure}
\usepackage{url}
\usepackage{verbatim}
\usepackage{algpseudocode}  
\usepackage{textcomp}
\usepackage{xcolor}
\usepackage{color}
\newenvironment{sequation}{\begin{equation}\small}{\end{equation}}

\newtheorem{theorem}{\textbf{Theorem}}

\usepackage{multirow}
\usepackage{bm}
\usepackage{booktabs}

\usepackage{epstopdf}
\usepackage{cases}
\usepackage{makecell,multirow,diagbox}
\usepackage{stfloats}
\usepackage{comment}
\usepackage{enumerate}

\definecolor{b}{rgb}{0.0, 0, 1}
\definecolor{k}{rgb}{0, 0, 0}
\definecolor{c}{rgb}{1, 0, 0}

\def\BibTeX{{\rm B\kern-.05em{\sc i\kern-.025em b}\kern-.08em
		T\kern-.1667em\lower.7ex\hbox{E}\kern-.125emX}}

\begin{document}
\pagestyle{empty} 
\title{\fontsize{20pt}{24pt}\selectfont J$\text{C}^5$A: Service Delay Minimization for Aerial MEC-assisted Industrial Cyber-Physical Systems}
	
\author{Geng~Sun,~\IEEEmembership{Senior Member,~IEEE},
        Jiaxu Wu,
        Zemin Sun,~\IEEEmembership{Member,~IEEE}, 
        Long He,
        Jiacheng Wang, \\
        Dusit Niyato,~\IEEEmembership{Fellow,~IEEE},
        Abbas Jamalipour,~\IEEEmembership{Fellow,~IEEE}, and 
        Shiwen Mao,~\IEEEmembership{Fellow,~IEEE}
	\thanks{This study is supported in part by the National Natural Science Foundation of China (62172186, 62272194), and in part by the Science and Technology Development Plan Project of Jilin Province (20230201087GX). (\textit{Corresponding author: Zemin Sun.)}
    \par Geng Sun is with the College of Computer Science and Technology, Jilin University, Changchun 130012, China, and also with the Key Laboratory of Symbolic Computation and Knowledge Engineering of Ministry of Education, Jilin University, Changchun 130012, China. He is also with the College of Computing and Data Science, Nanyang Technological University, Singapore 639798 (E-mail: sungeng@jlu.edu.cn).
    \par Jiaxu Wu, Zemin Sun, and Long He are with the College of Computer Science and Technology, Jilin University, Changchun 130012, China, and also with the Key Laboratory of Symbolic Computation and Knowledge Engineering of Ministry of Education, Jilin University, Changchun 130012, China (E-mails: wjx15941904130@163.com, sunzemin@jlu.edu.cn, and helong0517@foxmail.com).
    \par Jiacheng Wang, Dusit Niyato is with the School of Computer Science and Engineering, Nanyang Technological University, Singapore 639798 (E-mail: jiacheng.wang@ntu.edu.sg, dniyato@ntu.edu.sg).
    \par Abbas Jamalipour is with the School of Electrical and Computer Engineering, the University of Sydney, Sydney, NSW 2006, Australia (E-mail: a.jamalipour@ieee.org).
    \par Shiwen Mao is with the Department of Electrical and Computer Engineering, Auburn University, Auburn, USA (E-mail: smao@ieee.org).
    }
    }

\markboth{Journal of \LaTeX\ Class Files,~Vol.~14, No.~8, August~2015}%
{Shell \MakeLowercase{\textit{et al.}}: Bare Demo of IEEEtran.cls for Computer Society Journals}	

%
%
 \IEEEtitleabstractindextext{
\begin{abstract}
In the era of the sixth generation (6G) and industrial Internet of Things (IIoT), an industrial cyber-physical system (ICPS) drives the proliferation of sensor devices and computing-intensive tasks. To address the limited resources of IIoT sensor devices, unmanned aerial vehicle (UAV)-assisted mobile edge computing (MEC) has emerged as a promising solution, providing flexible and cost-effective services in close proximity of IIoT sensor devices (ISDs). However, leveraging aerial MEC to meet the delay-sensitive and computation-intensive requirements of the ISDs could face several challenges, including the limited communication, computation and caching (3C) resources, stringent offloading requirements for 3C services, and constrained on-board energy of UAVs. To address these issues, we first present a collaborative aerial MEC-assisted ICPS architecture by incorporating the computing capabilities of the macro base station (MBS) and UAVs. We then formulate a service delay minimization optimization problem (SDMOP). Since the SDMOP is proved to be an NP-hard problem, we propose a 
 \underline{\textbf{j}}oint 
 \underline{\textbf{c}}omputation offloading,  \underline{\textbf{c}}aching,  \underline{\textbf{c}}ommunication resource allocation, \underline{\textbf{c}}omputation resource allocation, and UAV trajectory \underline{\textbf{c}}ontrol \underline{\textbf{a}}pproach (J$\text{C}^5$A). Specifically, J$\text{C}^5$A consists of a block successive upper bound minimization method of multipliers (BSUMM) for computation offloading and service caching, a convex optimization-based method for communication and computation resource allocation, and a successive convex approximation (SCA)-based method for UAV trajectory control. Moreover, we theoretically prove the convergence and polynomial complexity of J$\text{C}^5$A. Simulation results demonstrate that the proposed approach can achieve superior system performance compared to the benchmark approaches and algorithms.
\end{abstract}
  
}

\maketitle
\IEEEdisplaynontitleabstractindextext
\IEEEpeerreviewmaketitle 
	%
	%
	\section{Introduction} 
	\label{sec:introduction}

\par Industry 5.0 marks a transformative era where human intelligence synergizes with advanced technologies to enhance the industrial efficiency and customization. Central to this shift is industrial cyber-physical systems (ICPS), which seamlessly integrate physical processes with computational and networking capabilities. Specifically, ICPS leverages the industrial Internet of Things (IIoT) to interconnect sensors, facilitating real-time monitoring, automatic control, and refined management across diverse industrial environments, which leads to an  exponential growth in multifarious sensor devices. Concurrently, the advancement of sixth generation (6G) technology has further propelled the proliferation of emerging mobile applications within ICPS, resulting in an increasing number of computation-intensive and delay-sensitive tasks such as virtual reality and online gaming. However, it is challenging to perform these tasks locally on IIoT sensor devices (ISDs) due to their constrained energy and computational capabilities and their physical sizes. To reduce the computational load on ISDs, cloud computing has been proposed as an effective solution. However, due to the geographical separation between the cloud infrastructure and ISDs, they may experience long communication latency.

\par Mobile edge computing (MEC)-assisted ICPS has been identified as a promising paradigm and has been extensively studied to offer low-latency offloading services close to ISDs. For example, Peng et al.~\cite{peng2022real} developed an MEC-assisted ICPS, where the terrestrial MEC nodes are deployed to support real-time transmission. Moreover, Ji et al.~\cite{ji2023intelligent} proposed an intelligent edge sensing and control framework by deploying MEC in the ICPS. However, these studies primarily focus on terrestrial MEC servers dependent on fixed ground infrastructures, which can result in frequent non-line-of-sight connections, high deployment costs, and limited environmental adaptability. This is particularly problematic in scenarios such as natural disasters, where deploying stationary infrastructures can be challenging.

\par To address the aforementioned challenges, the concept of unmanned aerial vehicle (UAV)-assisted MEC, i.e., aerial MEC, supported by features such as high maneuverability, flexibility, cost-effectiveness, and line-of-sight (LoS) connections, has been introduced to elevate the MEC facilities into the skies to enhance the flexibility of edge computing services. Consequently, by offloading computing tasks to the nearby UAVs, ISDs can flexibly enjoy cloud-like services and free themselves from the burden of heavy computation tasks anytime and anywhere. Recently, several studies investigated the aerial MEC-assisted ICPS. For example, Tang et al.~\cite{tang2023robust} considered a UAV-enabled ICPS, where a UAV is dispatched as an aerial edge server to assist IIoT data processing. Additionally, Shi et al.~\cite{Shi2024Deep} proposed an energy-constrained multi-UAV assisted ICPS, where multiple UAVs work together to adjust their frequencies based on the task sizes. However, most of these studies assume that the UAVs have relatively sufficient computing resources and can cache all services required by ISDs. However, this assumption may not be realistic due to the constrained resources of UAVs, especially in scenarios with dense offloading requirements. 

\par Fully exploring the benefits of UAV-assisted MEC to provide satisfactory offloading services for ICPS encounters significant challenges. \textbf{\textit{i) Resource management.}} In contrast to cloud servers with abundant resources, UAV-assisted MEC servers are typically equipped with limited communication, computation, and caching (3C) resources. However, the proliferation of ISDs, coupled with various delay-sensitive and computation-intensive applications, places unprecedented demands on 3C resources. Accordingly, the stringent 3C requirements of ISDs and the constrained 3C resources of MEC servers create difficulties in designing efficient 3C resource management to meet the demands of 3C-intensive tasks. \textbf{\textit{ii) Computing offloading.}} Different ISDs have heterogeneous offloading requirements for 3C services, while different MEC servers offer limited 3C resources. Moreover, random arrival of tasks lead to the spatiotemporal distribution of requirements, while the varying geographical deployment and capacities of different UAVs create a spatiotemporal distribution of 3C resources. Therefore, an effective computation offloading scheduling is crucial and challenging for the server load balancing. \textbf{\textit{iii) Trajectory control.}} While UAV-assisted MEC servers provide flexible 3C resources for computation offloading, a limited battery capacity of UAVs inherently restricts the service time, posing challenges for the energy-efficient UAV trajectory control.



\par The abovementioned challenges necessity efficient optimization of 3C resource management, computation offloading, and UAV trajectory control. However, focusing on just one aspect of these components is insufficient due to the following reasons. On the one hand, the optimization variables are mutually coupled. For example, the computation offloading decision depends on both the cached services and UAV locations. However, the limited storage capacity and energy of UAVs require efficient service caching and trajectory control. Moreover, the trajectory of a UAV also affects the communication quality and density of ISDs within its coverage, which in turn influences the communication and computing resource allocation. On the other hand, these optimization variables collectively determine the system performance. For example, offloading more tasks to nearby UAVs may reduce local computing delay but can lead to frequent caching and higher energy consumption for both computing and flying. Therefore, these interconnected optimization variables should be jointly optimized to achieve overall superior system performance, as it can effectively capture the intricate and coupling interactions and trade-offs among various optimization components. Consequently, we propose a collaborative optimization approach of computation offloading, service caching, communication resource allocation, computing resource allocation, and UAV trajectory control to enhance the performance of the ICPS. The main contributions of our work are summarized as follows:
\begin{itemize}
\item \textbf{\textit{Collaborative Aerial MEC-assisted Architecture.}} We propose a three-layer collaborative aerial MEC-assisted ICPS architecture, which consists of a set of ISDs, a collaborative UAV cluster, and a macro base station (MBS). Specifically, the UAVs act as the aerial MEC servers to collaboratively provide aerial computing services close to ISDs, and the MBS functions as the terrestrial MEC server to alleviate the overload of the UAV cluster.


\item  \textbf{\textit{Service Delay Optimization Problem Formulation.}} Considering the delay-sensitive requirements of the ISDs, we formulate a service delay minimization optimization problem (SDMOP). Specifically, the SDMOP aims to minimize the total delay of task completion under the energy constraints of ISDs and UAVs. Besides, we prove that this problem is an NP-hard and mixed integer nonlinear programming (MINLP) problem.

\item \textbf{\textit{Joint Optimization Approach.}} Since the formulated problem is difficult to be directly solved, we propose a \underline{\textbf{j}}oint \underline{\textbf{c}}omputation offloading,  \underline{\textbf{c}}aching,  \underline{\textbf{c}}ommunication resource allocation, \underline{\textbf{c}}omputing resource allocation, and UAV trajectory \underline{\textbf{c}}ontrol \textbf{a}pproach (J$\text{C}^5$A). J$\text{C}^5$A decomposes the original optimization problem into three subproblems, i.e., the computation offloading and service caching, the communication and computation resource allocation, and the UAV trajectory control. Specifically, for the subproblem of computation offloading and service caching, we employ the block successive upper bound minimization method of multipliers (BSUMM) to solve it. Moreover, the convex optimization methods are adopted to solve the subproblems of communication and computation resource allocation, and UAV trajectory control. Although we decompose the problem into three sub-problems to decouple the interdependent decision variables, we still achieve joint optimization, as a solution of one sub-problem will affect those of other sub-problems, thereby  preserving the joint benefits of considering these decisions within the same problem.

\item \textbf{\textit{Performance Validation.}} The effectiveness and performance of the
proposed J$\text{C}^5$A are validated through both theoretical
analysis and simulation experiment. We theoretically prove the convergence and polynomial complexity of J$\text{C}^5$A. Moreover, the simulation results demonstrate that J$\text{C}^5$A achieves superior performance than the comparative approaches and algorithms. 

\end{itemize}

\par The rest of this work is organized as follows. Section \ref{sec:Models and Preliminaries} presents the relevant models. Section \ref{sec:formulation} gives the problem formulation and analysis. Section \ref{sec:Algorithm} elaborates the proposed approach. Section \ref{sec:Simulation Results And Analysis} showcases the simulation results. Finally, the conclusions are presented in Section \ref{sec:Conclusion}.

\section{System Model}
\label{sec:Models and Preliminaries}

\par In this section, we illustrate the architecture of the collaborative aerial MEC-assisted ICPS and the system models.

%

\begin{figure}[t]
    \centering
    \setlength{\abovecaptionskip}{1pt}%
    \setlength{\belowcaptionskip}{1pt}%
    \includegraphics[width =3in]{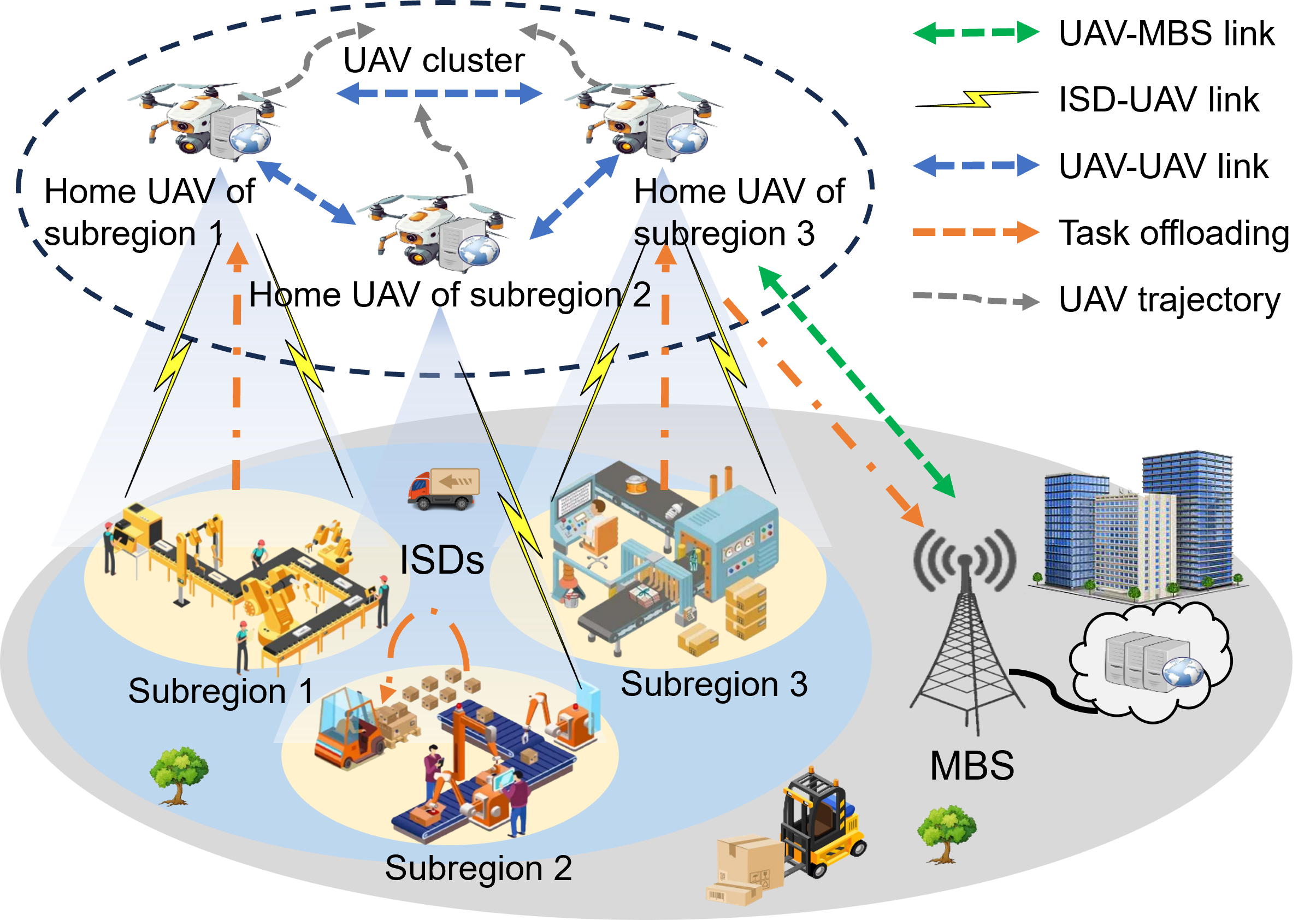}
    \caption{The collaborative aerial MEC-assisted ICPS.} 
\label{fig_SystemModel}
\end{figure}

\subsection{System Overview}
\label{subsec:System Model}

\par As shown in Fig. \ref{fig_SystemModel}, we consider a three-layer collaborative aerial MEC-assisted ICPS. The system consists of a device layer with $K$ ISDs, an aerial MEC layer with a cluster of $U$ UAVs, and a terrestrial MEC layer with an MBS. The system timeline is discretized into equal $N$ time slots, i.e., $\mathcal{N} = \{1,\ldots, n,\ldots, N\}$, where each slot duration $\tau$ is chosen to be sufficiently small such that each time slot can be considered to be quasi-static~\cite{Liu2022}.

\par \textit{At the device layer}, a set of ISDs $\mathcal{K}=\{1,\ldots, k,\ldots, K\}$ are randomly distributed in the considered area responsible for monitoring or performing the industrial activities and generate corresponding tasks. ISD $k$ is characterized by $(F_k^{\max}, E_k^{\max}, \boldsymbol{q}_k)$, where $F_k^{\max}$ denotes the local computing capability of ISD $k$, $E_k^{\max}$ indicates the energy constraint of ISD $k$, and $\boldsymbol{q}_k$ represents the horizontal position of ISD $k$. Moreover, we consider that each ISD could generate one computation task or not per time slot~\cite{Wang2022}, and the task $\boldsymbol{T}_k[n]$ generated by ISD $k$ in time slot $t$ can be characterized as 
\begin{sequation}
\label{equation_quadruples_of_taask}
    \boldsymbol{T}_k[n]=(d_k[n],s_k[n],c_k[n],t_k^{\max}), \ \forall k \in \mathcal{K},
\end{sequation}

\noindent where $d_k[n]$ denotes the task size (in bits), $s_k[n]$ indicates the required service type, $c_k[n]$ is the computational density (cycles/bit), and $t_k^{\max}$ represents the tolerable delay of the task. 

\par \textit{At the aerial MEC layer}, a set of rotary-wing UAVs $\mathcal{U}=\{1, \ldots, u,\ldots, U\}$ are deployed as aerial MEC servers\footnote{A UAV and an aerial MEC server will be used interchangeably.} to offer flexible edge computing services close to ISDs within the service area. To maximize the user capacity, the considered service area is initially divided into $U$ non-overlapping subregions, with each UAV responsible for one subregion. We denote $\mathcal{K}_u$ as the set of ISDs within the service area of UAV $u$, and refer to UAV $u$ as the home UAV of ISD $k\in \mathcal{K}_u$. Furthermore, these UAVs collaboratively form a cluster to share computing resources through horizontal computation offloading, with the resource allocation information stored in a shared resource allocation table~\cite{Ndikumana2020}. Moreover, each UAV $u\in\mathcal{U}$ is characterized by $(F_u^{\max},  E_u^{\max}, H_u, \boldsymbol{q}_u[n])$, wherein $F_u^{\max}$ denotes the computing capability of UAV $u$, $E_u^{\max}$ represents the energy constraint of UAV $u$, $C_u^{\max}$ is the service caching storage of UAV $u$, $H$ indicates the altitude of UAV $u$, and $\boldsymbol{q}_u[n]$ means the horizontal position of UAV $u$ at time slot $[n]$. We consider that each UAV $u\in\mathcal{U}$ flies at a fixed altitude to mitigate additional energy consumption associated with the frequent altitude changes~\cite{li2023robust}. In addition, the trajectory of each UAV should satisfy the physical constraints as follows: 
\begin{subequations}
\label{eq_UAV}
\begin{alignat}{2}
    &||\mathbf{q}_u[n+1]-\mathbf{q}_u[n]|| \leq V^{\max} \tau, \forall u, n,  \label{AAi}\\
    &||\mathbf{q}_u[n]-\mathbf{q}_v[n]|| \geq D^{\min},  \forall u \neq v, n, \label{AAj}\\    &\mathbf{q}_u[0]=\mathbf{q}_u^I , \forall u,  \label{AAk}
\end{alignat}
\end{subequations}

\noindent where $V^{\max}$ denotes the maximum velocity of UAVs, $D^{\min}$ represents the safe distance between UAVs, and $\mathbf{q}_u^I$ is the initial position of UAV $u$. Moreover, (\ref{AAi}) means that the flight distance of each UAV is constrained by the maximum velocity, (\ref{AAj}) guarantees the safe distance between UAVs, and (\ref{AAk}) limits the initial position of UAV $u$.


\par \textit{At the terrestrial MEC layer}, an MBS $M$ is connected to a terrestrial MEC server\footnote{An MBS and a terrestrial MEC server will be used interchangeably.} via a wired optical fiber link to provide stronger computing capabilities. Specifically, MBS $M$ is characterized by $(F_M^{\max},  E_M^{\max}, H_M, \boldsymbol{q}_M)$, wherein $F_M^{\max}$ represents the computing resources of the MBS, $E_M^{\max}$ denotes the energy constraint of the MBS, $H_M$ is the height of the MBS, and $\boldsymbol{q}_M$ is the horizontal position of the MBS.


\par The workflow of the considered ICPS can be described as follows. For task $\boldsymbol{T}_k[n]$ generated by ISD $k\in \mathcal{K}$ in time slot $n\in \mathcal{N}$, it can be executed locally at the ISD (i.e., local computing), offloaded to a UAV for execution (i.e., UAV-assisted computing), or offloaded to the MBS for execution (i.e., MBS-assisted computing)\footnote{The UAVs and the MBS are collectively referred to MEC servers.}. For local computing, the ISD performs the task based on its local computing capabilities. Moreover, for the UAV-assisted computing, the task is initially offloaded to the home UAV of the ISD, and the home UAV then allocates computing resources to execute the task or offloads the task to other UAVs within the UAV cluster for execution. Additionally, for the MBS-assisted computing, the task is relayed from the home UAV to the MBS for execution.

\subsection{Communication Model }
\label{sec_communication_model}

\par In the considered ICPS, three types of communication links are presented, i.e., ISD-UAV link, UAV-UAV link, and UAV-MBS link. Moreover, the orthogonal frequency-division multiple access (OFDMA) technique is employed for these communications to mitigate unreliable transmissions caused by interference. The aforementioned communication links are described as follows.

\subsubsection{ISD-UAV Link}
\par Considering the communication between ISDs and UAVs is dominated by LoS link \cite{Zhou2021}, the channel power gain between ISD $k$ and its home UAV $u$ in time slot $n$ can be given as
\begin{sequation}
    G_{k,u}[n]=\beta_0/(\|\boldsymbol{q}_k-\boldsymbol{q}_u[n]\|^2+H^2),
\end{sequation}

\noindent where $\beta_0$ denotes the channel power gain at a reference distance of 1 meter. Therefore, the data transmission rate from ISD $k$ to UAV $u$ can be calculated as
\begin{sequation}
    R_{k,u}[n]=\theta_{k,u}[n]B_{f}\log_2\big(1+P_{k}G_{k,u}[n]/\sigma^2\big),
\end{sequation}

\noindent where $B_{f}$ denotes the total bandwidth available for the ISD-UAV link, $P_{k}$ represents the transmit power of ISD $k$, $\sigma^2$ means the noise power, and $\theta_{k,u}[n]$ is the bandwidth allocation coefficient.

\subsubsection{UAV-UAV Link}

\par By considering that the communication between UAVs is dominated by the LoS link, the channel power gain between UAV $u$ and UAV $v$ can be given as
\begin{sequation}
    G_{u,v}[n]=\beta_0/\|\boldsymbol{q}_u[n]-\boldsymbol{q}_v[n]\|^2. 
\end{sequation}

\noindent Therefore, the data transmission rate from UAV $u$ to UAV $v$ can be obtained as
\begin{equation}
    R_{u,v}[n]=\theta_{u,v}[n]B_{c}\log_2(1+P_{u}G_{u,v}[n]/\sigma^2),
\end{equation}
\noindent where $B_{c}$ denotes the total bandwidth for the UAV-UAV link, $P_{u}$ is the transmit power of UAV $u$, and $\theta_{u,v}[n]$ represents the fraction of bandwidth which is allocated to the communication between UAV $u$ and UAV $v$.

\subsubsection{UAV-MBS Link}

\par The channel power gain between UAV $u$ and MBS $M$ is given as
\begin{sequation}
    G_{u,M}[n]=\beta_0/(\|\boldsymbol{q}_u[n]-\boldsymbol{q}_M\|^2+(H_u-H_M)^2).
\end{sequation}

\noindent Then the transmission rate between UAV $u$ and MBS $M$ can be calculated as
\begin{sequation}
\label{eq.R_u_M}
    R_{u,M}[n]=\theta_{u,M}[n]B_{b}\log_2(1+P_{u}G_{u,M}[n]/\sigma^2),
\end{sequation}

\noindent where $B_{b}$ denotes the total bandwidth available in the UAV-MBS link, and $\theta_{u,M}[n]$ is the fraction of bandwidth allocated to the communication between UAV $u$ and MBS $M$.

\subsection{Caching Model}
\label{sec_caching_model}

\par The service library of the MBS is denoted by $\mathcal{S}=\{1,\ldots,s,\ldots, S\}$, where the cache contents are divided into several units of the same size~\cite{Ji2020}. Moreover, the service caching of UAVs is initialized based on the ranking of content popularity. Specifically, the least recently used caching replacement policy is adopted to update the cache of UAVs from the MBS~\cite{Lee2001}.  Note that if UAV $u$ does not cache the content of service $s_k$, it is unable to process task $\boldsymbol{T}_k[n]$. Therefore, we introduce an additional constraint as follows:
\begin{sequation}
\label{eq_con_caching}
    o_k^u[n](1-z_k^u[n])=0, \forall k, u, n,
\end{sequation}

\noindent where the binary variables $o_k^u[n]\in\{0,1\}$ and $z_k^u[n]\in\{0,1\}$ denote the offloading decision for task $\boldsymbol{T}_k[n]$ and the caching decision of UAV $u$ in time slot $n$, respectively. Moreover, Eq. \eqref{eq_con_caching} implies that if task $\boldsymbol{T}_k[n]$ is offloaded to UAV $u$, i.e., $o_k^{u}[n]=1$, the service required by this task must have been cached by the UAV, i.e., $z_k^u[n]=1$.

\subsection{Service Delay Model}
\label{sec_service_delay}

\par The service delay for each task depends on the offloading decision. Specifically, the task $\boldsymbol{T}_k[n]$ can be processed locally on the ISD $k$ (referred to as local computing), directly offloaded to the home UAV $u$ (referred to as home UAV-assisted computing), transmitted from the home UAV $u$ to the MBS $M$ (referred to as MBS-assisted computing), or transmitted from the home UAV $u$ to the neighbor UAV $v$ (referred to as UAV collaborative computing). To this end, we define a set of binary variables $\{o^k_k[n], o^k_u[n],o^k_v[n],o^k_M[n]\} $ to represent the offloading decision of ISD $k$ at time slot $n$, where $o_k^k[n]=1$, $o_k^u[n]=1$, $o_k^v[n]=1$, and $o_k^M[n]=1$ indicate that the task is processed locally, offloaded to the home UAV $u$, offloaded to a neighboring UAV $v$, and offloaded to the MBS $M$, respectively. Note that for edge computing, we ignore the result feedback delay since the results of most mobile applications is typically much smaller than the input data~\cite{Wang2022}.

\subsubsection{Local Computing}

\par  When task $\boldsymbol{T}_k[n]$ is processed by ISD $k$ locally, the service delay can be given as follows:
\begin{sequation}
\label{eq_t_k_loc}
    t_k^{\text{loc}}[n]=d_k[n] c_k[n]/F_k^{\max}.
\end{sequation}

\noindent where $F_k^{\max}$ is the local computing capability of ISD $k$.

\subsubsection{Home UAV-assisted Computing}

\par When task $\boldsymbol{T}_k[n]$ is offloaded from ISD $k$ to the home UAV $u$ where the required service is cached, the service delay can be given as
\begin{sequation}
\label{eq_t_k_u}
    t_k^u[n] = d_k[n]/R_{k,u}[n]+d_k[n] c_k[n]f_k^u[n],
\end{sequation}

\noindent where $f_k^u[n]$ denotes the amount of computation resource allocated by the home UAV $u$ to task $\boldsymbol{T}_k[n]$ in time slot $n$.



\subsubsection{UAV Collaborative Computing}

\par When the home UAV $u$ is unable to handle task $\boldsymbol{T}_k[n]$, it checks the resource allocation table and offloads the task to UAV $v$ with available resources in the same UAV cluster. In this case, the service delay can be given as
\begin{sequation}
    \label{eq_t_k_v}
    t_k^v[n] = d_k[n]/R_{k,u}[n]+d_k[n]/R_{u,v}[n]+d_k[n] c_k[n]/f_k^v[n],
\end{sequation}


\noindent where $f_k^v[n]$ represents the amount of computation resource allocated by UAV $v$ to task $\boldsymbol{T}_k[n]$ in time slot $n$.

\subsubsection{MBS-assisted Computing}

\par If task $\boldsymbol{T}_k[n]$ is offloaded to the MBS $M$, the service delay can be calculated as
\begin{sequation}
\label{eq_t_k_M1}
    t_k^M[n] = d_k[n]/R_{k,u}[n]+d_k[n]/R_{u,M}[n]+d_k[n] c_k[n]/f_k^M[n],
\end{sequation}

\noindent where $f_k^M[n]$ is the amount of computation  resource allocated by MBS $M$ to task $\boldsymbol{T}_k[n]$ in time slot $n$.



\subsubsection{Total Service Delay}

\par Therefore, the total delay for processing task $\boldsymbol{T}_k[n]$ can be obtained as

\vspace{-0.8em}
{\small 
\begin{align}
  & t_k[n] =  \underbrace{o_k^k[n]t_k^{\text{loc}}[n]}_{\text{Local computing }}+\underbrace{o_k^u[n]t_k^u[n]}_{\text{Home UAV-assisted computing}} + \underbrace{\sum\limits_{\substack{v\in\mathcal{U},v\neq u}}o_k^v[n]t_k^v[n]}_{\text{UAV collaborative computing }}\notag\\&+\underbrace{o_k^M[n]t_k^M[n]}_{\text{MBS-assisted computing}}+\sum\limits_{u\in\mathcal{U}}\underbrace{\max\{z_k^u[n] - z_k^u[n-1],0\} D_{s}/\bar{R}}_{\text{Caching}}, 
\end{align}
}

\noindent where $\max\{z_k^u[n] - z_k^u[n-1],0\} \frac{D_{s}}{\bar{R}}$ denotes the backhaul delay for UAV $u$ to cache the service $s_k$ from the MBS, and $\bar{R}$ is the average transmission rate from the MBS to a UAV. 

\vspace{-0.6em}

\subsection{Energy Consumption Model}
\label{sec_energy_consumption}

\par The energy consumption of ISDs and UAVs depends on the offloading decision, which is detailed as follows.

\subsubsection{Energy consumption of ISDs}

\par The energy consumption of ISD $k$ can be calculated as
\begin{sequation}
\begin{aligned}
\label{eq_t_k_M}
    E_k[n] = \underbrace{o_k^k[n]E_{k}^{\text{com}}[n]}_{\text{Computation energy}}+\underbrace{(o_k^u[n]+\sum\limits_{\substack{v\in\mathcal{U},v\neq u}}o_k^v[n]+o_k^M[n])E_{k,u}^{\text{tran}}[n]}_{\text{Transmission energy}}.
\end{aligned}
\end{sequation}

\noindent First, $E_{k}^{\text{com}}[n]=\varpi_k(F_k^{\max})^{2}d_k[n] c_k[n]$ represents the energy consumption of ISD $k$ for local computing, where $\varpi_k$ is the effective switched capacitance coefficient of ISD $k$~\cite{PanPWZW21}. Moreover, $E_{k,u}^{\text{tran}}=P_kd_k[n]/R_{k,u}[n]$ means energy consumption of ISD $k$ for task uploading.
 
\subsubsection{Energy consumption of a UAV}

\par The energy consumption for UAV $u$ can be calculated as
\begin{sequation}
\begin{aligned}
\label{eq_E_u}
    E_u[n] &= \sum_{k\in\mathcal{K}}\big(\underbrace{o_k^u[n] E_u^{\text{com}}[n]}_{\text{Computation energy}} + \sum_{\substack{v\in\mathcal{U},v\neq u}}\underbrace{o_k^v[n]E_{u,v}^{\text{tra}}[n]}_{\text{Transmission energy to UAV}}\\&+\underbrace{o_k^M[n]E_{u,M}^{\text{tra}}[n]}_{\text{Transmission energy to MBS}}\big)+\underbrace{E_u^{\text{pro}}\tau}_{\text{Flying energy}}.
\end{aligned}
\end{sequation}

\noindent First, $E_u^{\text{com}}[n]=\varrho_ud_k[n]c_k[n]$ is the energy consumption of UAV $u$ to process task $\boldsymbol{T}_k[n]$, where $\varrho_u$ denotes the energy consumption per unit CPU cycle of UAV $u$. Moreover, $E_{u,v}^{\text{tra}}[n]=d_k[n]/R_{u,v}[n]$ and $E_{u,M}^{\text{tra}}[n]=d_k[n]/R_{u,M}[n]$ denote the energy consumption of UAV $u$ when transmitting the task to UAV $v$ and MBS $M$, respectively, for further processing.  Besides, $E_u^{\text{fly}}[n]=E_u^{\text{pro}}\tau$ means the unit propulsion energy of ISD $k$ which can be given as~\cite{zeng2019accessing} $E_u^{\text{pro}}=\vartheta_1(1+3||\mathbf{v}_u[n]||^2/({v_u^{\text{tip}}})^2)+\vartheta_2 \sqrt{\sqrt{ \vartheta_3+||\mathbf{v}_u[n]||^4/4}-||\mathbf{v}_u[n]||^2/2}+ \vartheta_4||\mathbf{v}_u[n]||^3$, where $\mathbf{v}_u[n]=||\boldsymbol{q}_u[n+1]-\boldsymbol{q}_u[n]||/\tau$ is the instantaneous velocity of UAV $u$, $v_u^{\text{tip}}$ denote the tip speed of the rotor blade, and $ \vartheta_1$, $\vartheta_2$, $ \vartheta_3$, and $ \vartheta_4$ are the constants that depend on the aerodynamic properties.

\subsubsection{Energy consumption of MBS}
\par The energy consumption for MBS $M$ can be given as
\begin{sequation}
    \label{eq_MBS_energy}
    E_M[n] = \sum_{k\in\mathcal{K}}o_k^M[n]\varrho_Md_k[n]c_k[n],
\end{sequation}

\noindent where $\varrho_M$ represents the energy consumption per unit CPU cycle of MBS $M$.

\section{Problem Formulation and Analysis}
\label{sec:formulation}

\par This section shows the problem formulation and analysis.

%
%
\subsection{Problem Formulation}
\label{sec:Problem Formulation}

\par We aim to minimize the service delay in each time slot through jointly optimizing the computation offloading $\mathbf{O} = \{o_k^k[n],o_k^u[n],o_k^v[n],o_k^M[n]\}_{\forall k, u, n}$, service caching $\mathbf{Z} = \{z_k^u[n]\}_{\forall k, u, n}$, communication resource allocation $\mathbf{\Theta} = \{\theta_{k,u}[n],\theta_{u,v}[n],\theta_{u,M}[n]\}_{\forall k, u,, n}$, computation resource allocation $\mathbf{F} = \{f_k^u[n],f_k^v[n], f_k^M[n]\}_{\forall k, u, n}$, and UAV trajectory control $\mathbf{Q} = \{\boldsymbol{q}_u[n]\}_{\forall u,n}$. Our problem is formulated as a single time-slot optimization to obtain the optimal decisions for each time slot, allowing the approach to adapt to the dynamic characteristics of the ICPS such as the resource availability, workload, and UAV trajectory. Thus, the SDMOP can be mathematically formulated as follows:

\vspace{-0.8em}
{\small
\begin{align}
    \mathbf{P}: \ &{\min_{\mathbf{O,Z,Q,\Theta,F}} \ \sum\limits_{k\in\mathcal{K}}t_k[n]}, \label{AA} \\
    \text{s.t.}\ \ 
    &\sum\limits_{k \in \mathcal{K}_u} \theta_{k,u}[n](1-o_k^k[n]) \leq 1, \forall u, n, \tag{\ref{AA}{a}} \label{AAa}\\
    &\sum\limits_{u,v\in\mathcal{U},u \neq v} \theta_{u,v}[n]o_k^v[n]\leq 1, \forall n, \tag{\ref{AA}{b}} \label{AAb}\\
    &\sum\limits_{u\in\mathcal{U}} \theta_{u,M}[n]o_k^M[n]\leq 1, \forall k, n, \tag{\ref{AA}{c}} \label{AAc}\\
    &\sum\limits_{k \in \mathcal{K}} o_k^u[n]   f_k^u[n] \leq F_u^{max}, \forall u,n ,\tag{\ref{AA}{d}} \label{AAd}\\
    &\sum\limits_{k \in \mathcal{K}}o_k^M[n] f_k^M[n] \leq F_M^{max}, \forall n,\tag{\ref{AA}{e}} \label{AAe}\\
    &\sum\limits_{k \in \mathcal{K}} z_k^u[n] \leq C_u^{\max}, \forall u,n,
    \tag{\ref{AA}{f}} \label{AAf}\\ &o_k^k[n]+o_k^u[n]+\sum\limits_{\substack{v\in\mathcal{U},v \neq u}}o_k^v[n]+o_k^M[n]=1,  \forall k,u,n,
    \tag{\ref{AA}{g}} \label{AAg}\\
    & E_k[n]\leq E_k^{\max}, \forall k,n, \tag{\ref{AA}{h}}\label{AAf1}\\
    & E_u[n]\leq E_u^{\max}, \forall u,n, \tag{\ref{AA}{i}}\label{AAf2}\\
    & E_M[n]\leq E_M^{\max}, \forall n, \tag{\ref{AA}{j}}\label{AAEM}\\
    &t_k[n]<t_k^{\max}, \forall k, \forall n,\tag{\ref{AA}{k}} \label{AAl}\\
    &\eqref{AAi}-\eqref{AAk}, \eqref{eq_con_caching}, \notag 
\end{align}}

\noindent where (\ref{AAa}), (\ref{AAb}), and (\ref{AAc}) represent the constraints of bandwidth allocation. Moreover, (\ref{AAd}) and (\ref{AAe}) constrain the computation resources of UAVs and the MBS. Then, (\ref{AAf}) constrains the caching resources for UAVs, and (\ref{AAg}) guarantees the binary decision of computation offloading. Furthermore, constraints (\ref{AAf1}), (\ref{AAf2}), and (\ref{AAEM}) impose limits on the energy consumption of ISD $k$, UAV $u$, and MBS $M$ respectively. In addition, \eqref{AAl} ensures that each task can be completed within the deadline. Besides, \eqref{AAi} to \eqref{AAk} constrain the UAV trajectories, and \eqref{eq_con_caching} constrain the relationship between computation offloading and service caching.

 
\subsection{Problem Analysis}

\subsubsection{Challenges} 

\par Solving problem $\mathbf{P}$ directly may presents several challenges as follows: \textit{i) NP-hard and non-convex.} Problem $\mathbf{P}$ involves both binary decision variables (i.e., computation offloading $\mathbf{O}$ and service caching $\mathbf{Z}$) and continuous decision variables (i.e., resource allocation $\{\mathbf{\Theta},\mathbf{F}\}$ and UAV trajectory control $\mathbf{Q}$), which is an NP-hard and non-convex MINLP~\cite{Boyd2014}. \textit{ii) Interdependent decision variables.} The decision variables of problem $\mathbf{P}$ are coupled and interdependent with each other. On the one hand, these decision variables collectively affect the service delay. On the other hand, these decision variables are mutually dependent. For example, the decision of service caching can influence the resource allocation, which in turn impact the UAV trajectory control. \textit{iii) Extensive decision spaces.} Problem $\mathbf{P}$ not only involves diverse decisions ($\mathbf{O,Z,Q,\Theta,F}$), but also consists of multiple types of decisions for computation offloading ($o^k_k[n], o^k_u[n],o^k_v[n],o^k_M[n]$), communication resource allocation $\{\theta_{k,u}[n],\theta_{u,v}[n],\theta_{u,M}[n]\}$, and computation resource allocation $\{f_k^u[n],f_k^v[n], f_k^M[n]\}$, which leads to large decision spaces, especially in dense scenarios.



\subsubsection{Motivation}

\par The aforementioned challenges motivate the design of a low-complexity and high-efficiency approach. To effectively address these challenges, we propose the J$\text{C}^5$A, guided by the following key motivations:

\begin{enumerate}[i)]
\item \textit{Challenges of NP-hard and non-convex MINLP in $\mathbf{P}$.} The NP-hard and non-convex MINLP feature of problem $\mathbf{P}$ makes it challenging to directly find an optimal solution. Although the machine learning algorithms such as deep reinforcement learning (DRL) are effective for decision making, the mixed integral, coupled, and extensive decision spaces of problem $\mathbf{P}$ could lead to complex action spaces. Therefore, DRL would require extensive environmental interactions and experience long convergence time to acquire optimal policies, making it costly in the resource-constrained aerial MEC-assisted ICPS.

\item \textit{Decoupling interdependent decision variables.} The interdependent decision variables motivate us to decouple them by decomposing problem $\mathbf{P}$ into manageable subproblems~\cite{Chu2024Joint}. By solving each subproblem individually, we can simplify the decision making process while satisfying the task offloading requirements of ISDs and meeting the resource constraints of the MEC servers.

\item \textit{Balancing complexity reduction and performance in extensive decision spaces.} When decoupling the problem, the extensive decision space poses a challenge for achieving the tradeoff between the complexity reduction and performance degradation. Specifically, while dividing the problem into multiple subproblems can lower computational complexity, it may lead to performance degradation due to the narrower solution space. Conversely, solving fewer subproblems may yield higher solution quality but complicate the problem solving.  Therefore, the problem division of the proposed approach is driven by the need to reduce the problem complexity while maintaining the satisfied performance.  


\end{enumerate}

\section{The proposed J$\text{C}^5$A}
\label{sec:Algorithm}

\par  We propose J$\text{C}^5$A to solve the formulated SDMOP in this section. Specifically, J$\text{C}^5$A first decouples problem $\mathbf{P}$ into three subproblems, i.e., computation offloading and service caching, communication and computation resource allocation, and UAV trajectory control. Then, by iteratively solving the subproblems, we finally obtain a high-performance suboptimal solution. Fig.~\ref{fig_algrithm-framework} shows the J$\text{C}^5$A framework.

\begin{figure*}[h]
    \centering
    \setlength{\abovecaptionskip}{0pt}%
    \setlength{\belowcaptionskip}{0pt}%
    \includegraphics[width =6.8in]{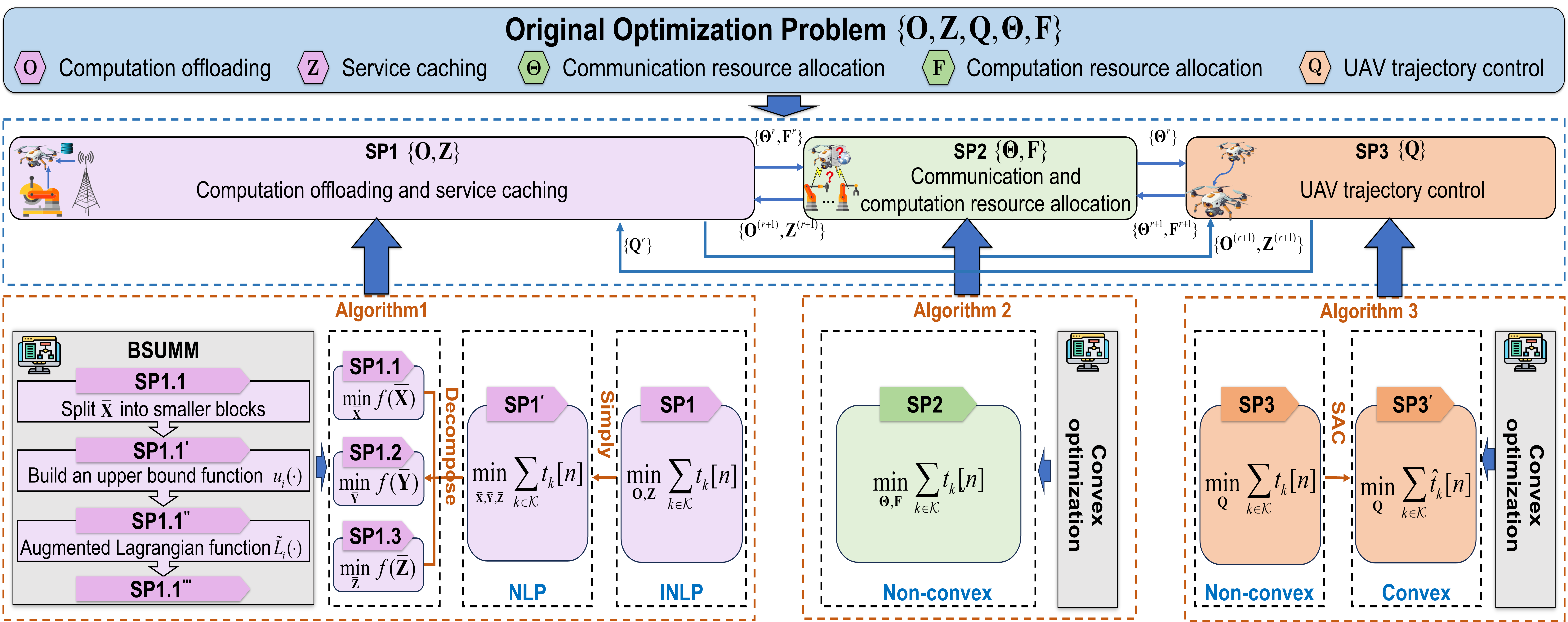}
    \caption{The framework of J$\text{C}^5$A.} 
\label{fig_algrithm-framework}
\vspace{-1.3em}
\end{figure*}

\subsection{Computation Offloading and Service Caching}

\par Given communication resource allocation $\hat{\mathbf{\Theta}}$, computation resource allocation $\hat{\mathbf{F}}$, and UAV trajectory control $\hat{\mathbf{Q}}$, problem $\mathbf{P}$ can be transformed into subproblem $\mathbf{SP1}$ to determine the computation offloading $\mathbf{O}$ and service caching $\mathbf{Z}$, which is reformulated as follows:

 \vspace{-0.8em}
{\small
\begin{align}
\label{SP1}
    \mathbf{SP1}: \ &\min_\mathbf{O,Z} \ \sum\limits_{k\in\mathcal{K}}t_k[n],\\
    \text{s.t.} \ &(\ref{AAa})-(\ref{AAl}), \eqref{eq_con_caching}. \notag
\end{align}}

\begin{theorem}
\label{th_SP1}
    Subproblem $\mathbf{SP1}$ is an integer non-linear programming problem (INLP).
\end{theorem}
\begin{proof}
    \label{proof_sp1_NP}
     First, the subproblem $\mathbf{SP1}$ only involves binary variables. Furthermore, the objective function \eqref{SP1} contains a nonlinear term, i.e., $\max\{z_k^u[n] - z_k^u[n-1],0\} D_{s_k}/\bar{R}$. Therefore, subproblem $\mathbf{SP1}$ is an INLP.
\end{proof}

\par Theorem \ref{proof_sp1_NP} indicates that solving $\mathbf{SP1}$ directly is challenging. Specifically, subproblem $\mathbf{SP1}$ has two significant characteristics. On the one hand, $\mathbf{SP1}$ involves two types of decisions, i.e., the task offloading decision made by ISDs ($\mathbf{O}$) and the service caching decision made by UAVs ($\mathbf{Z}$)). On the other hand, the task offloading decision include both the local computing determined by ISD resources ($o^k_k[n]$) and the edge offloading decision determined by MEC server resources ($o^k_u[n],o^k_v[n],o^k_M[n]$). Based on these characteristics, the decisions of $\mathbf{SP1}$ can be categorized separately for distributed and parallel decision-making to improve the solving efficiency. Therefore, we solve the $\mathbf{SP1}$ by the following steps. First, we simplify \eqref{SP1} by dividing the computation offloading decision into the decisions of local computing and edge offloading. Then, considering that the distributed BSUMM method can solve the problem in parallel, offering advantages in solution speed and decomposability~\cite{han2017signal}, we employ the BSUMM to solve the simplified $\mathbf{SP1}$.

\subsubsection{Simplification for Subproblem $\mathbf{SP1}$} 

\par To achieve the distributed parallel computing, the variable of computation offloading $\mathbf{O}$ is divided into the parts of local computing and edge offloading, which are as follows:

\vspace{-0.8em}
{\small
\begin{subequations}
    \begin{alignat}{1}
        &\mathbf{X}=\big\{x_k[n]\bigg|x_k[n]=1-o_k^k[n],\forall k, n\big\},\label{X_transform}\\
        &\begin{aligned}
            \mathbf{Y}=\big\{&y_k^u[n],y_k^v[n],y_k^M[n]\big|y_k^u[n]=o_k^u[n],\\
            &y_k^v[n]=o_k^v[n],y_k^M[n]=o_k^M[n] ,\forall k,u\neq v,n \big\},
        \end{aligned}
    \end{alignat}
\end{subequations}}
\noindent where \eqref{X_transform} eliminates the constant in constraint \eqref{AAg}.

\par {\textbf{First}}, (\ref{AAg}) is transformed into as follows:

\vspace{-0.8em}
{\small
\begin{subequations}
\label{x_y_1_2}
    \begin{alignat}{1}
        &x_k[n] \geq y_k^u[n] + \sum\limits_{\substack{v\in\mathcal{U},v \neq u}}y_k^v[n] + y_k^M[n],\forall k, u, n, \label{x_y_1}\\
        &x_k[n] \leq \max\{y_k^u[n],y_k^v[n],y_k^M[n]\}, \forall k, u, n. \label{x_y_2}
    \end{alignat}
\end{subequations}}

\par {\textbf{Second}}, it can be deduced that if $y_k^u[n] = 1$, then $z_k^u[n]=1$, and if $y_k^u[n] = 0$, then $z_k^u[n] \in \{0,1\}$. Therefore, constraint (\ref{eq_con_caching}) can be rewritten as follows:
\begin{sequation}
\label{y_z}
    y_k^u[n] \leq z_k^u[n], \forall k, u, n. 
\end{sequation}

\par {\textbf{Finally}}, we calculate the sum of constraints with respect to the computation offloading $\mathbf{O}=\{\mathbf{X},\mathbf{Y}\}$ and service caching $\mathbf{Z}$. Specifically, we relax these variables into continuous variables, which are as follows:

 \vspace{-0.8em}
{\small
\begin{subequations}
    \label{eq_binary_to_continuous}
    \begin{alignat}{1}
        &\overline{\mathbf{X}}\triangleq\big\{\sum_{k\in\mathcal{K}} x_k[n]=1,x_k[n]\in[0,1]\big\},\label{overline_X}\\
        &\begin{aligned}                              
        \overline{\mathbf{Y}}\triangleq\big\{\sum_{u\in\mathcal{U}}\sum_{k\in\mathcal{K}_u}y_k^u[n]+\sum\limits_{\substack{v\in\mathcal{U}, v\neq u}} y_k^v[n]+y_k^M[n]=1,\\
            y_k^u[n],y_k^v[n],y_k^M[n]\in[0,1]\big\} ,
        \end{aligned}\label{overline_Y}\\
        &\overline{\mathbf{Z}}\triangleq \big\{ \sum_{u\in\mathcal{U}}\sum_{k\in\mathcal{K}} z_k^u[n] = 1, z_k^u[n] \in [0,1] \big\}.\label{overline_Z}
    \end{alignat}
\end{subequations}}

\par Based on the above analysis, the subproblem $\mathbf{SP1}$ can be transformed as follows:

\vspace{-0.8em}
{\small
\begin{align}
\mathbf{SP1^{\prime}}:\ &\min_\mathbf{\overline{X},\overline{Y},\overline{Z}} \ {\sum\limits_{k\in\mathcal{K}}t_k[n]}, \label{SP1_pie}\\
    \text{s.t.} \ \ &\sum\limits_{k \in \mathcal{K}_u} \theta_{k,u}[n]x_k[n]\leq 1, \forall u, n \tag{\ref{SP1_pie}{a}} \label{SP1_pie_a}\\
    &\sum\limits_{u,v\in\mathcal{U}, u \neq v,} \theta_{u,v}[n]y_k^v[n]\leq 1, \forall k, n, \tag{\ref{SP1_pie}{b}} \label{SP1_pie_b}\\
    &\sum\limits_{u\in\mathcal{U}} \theta_{u,M}[n]y_k^M[n]\leq 1, \forall k, n, \tag{\ref{SP1_pie}{c}} \label{SP1_pie_c}\\
    &\sum\limits_{k \in \mathcal{K}} y_k^u[n]   f_k^u[n] \leq F_u^{max}, \forall u, n ,\tag{\ref{SP1_pie}{d}} \label{SP1_pie_d}\\
    &\sum\limits_{k \in \mathcal{K}}y_k^M[n] f_k^M[n] \leq F_M^{max}, \forall n ,\tag{\ref{SP1_pie}{e}} \label{SP1_pie_e}\\
    &\eqref{AAf},\ \eqref{AAl},  \ \text{and} \ \eqref{x_y_1}-\eqref{overline_Z} \notag.
\end{align}}

\par Solving subproblem $\mathbf{SP1^{\prime}}$ directly remains challenging since it is a  non-linear problem (NLP), as given in Theorem \ref{th_SP1_pie}. Therefore, we will present the solution of $\mathbf{SP1^{\prime}}$ in the next subsection.

\begin{theorem}
\label{th_SP1_pie}
    Subproblem $\mathbf{SP1^{\prime}}$  is a convex NLP.
\end{theorem}

\begin{proof}
    Problem $\mathbf{SP1^{\prime}}$ only involves continuous variables (i.e., $\mathbf{\overline{X}}$, $\mathbf{\overline{Y}}$, and $\mathbf{\overline{Z}}$), and it includes linear relationships with respect to these variables except for two terms, i.e., $\max\{y_k^u[n],y_k^v[n],y_k^M[n]\}$ in \eqref{x_y_2} and $\max\{z_k^u[n] - z_k^u[n-1],0\}$ in \eqref{SP1_pie}. Therefore, $\mathbf{SP1^{\prime}}$ is an NLP~\cite{Chu2024Joint}. Moreover, it can be proved that $\omega(a,b)=\max\{a,b\}$ is convex with respect to $a$ and $b$~\cite{boyd2004convex}. Specifically, $\forall (a_1,b_1), \ (a_2,b_2) \in \mathbb{R}^2$ and $\delta \in [0,1]$, we can deduce that
    \begin{sequation}
    \begin{aligned}
        &\omega\left(\delta a_1+(1-\delta)a_2,\delta b_1+(1-\delta)b_2\right)\\
        = \ &\max\{\delta a_1+(1-\delta)a_2,\delta b_1+(1-\delta)b_2\}\\
        \leq \  &\delta \max\{a_1,b_1\}+(1-\delta)\max\{a_2,b_2\}\\
        = \ &\delta \omega(a_1,b_1)+(1-\delta)\omega(a_2,b_2),
    \end{aligned}
    \end{sequation}
    
    \noindent where $\omega(a,b)$ is convex. Therefore, subproblem $\mathbf{SP1^{\prime}}$ is a convex NLP.
\end{proof}

\subsubsection{Solution for Problem $\mathbf{SP1^{\prime}}$} 

\par The BSUMM is employed to solve subproblem $\mathbf{SP1^{\prime}}$ in a distributed manner. Specifically, we iteratively optimize $\overline{\mathbf{X}}$, $\overline{\mathbf{Y}}$, and $\overline{\mathbf{Z}}$ by splitting $\mathbf{SP1^{\prime}}$. For the convenience of description, $f(\cdot)$ is used to denote the objective function of each subproblem as follows: 

\vspace{-0.8em}
{\small
\begin{alignat}{1}
    \mathbf{SP1.1}:\quad&\min_{\mathbf{\overline{X}}} \ f(\mathbf{\overline{X}}), \label{SP1.1}\\
    \text{s.t.} \ \ &\eqref{AAl},\eqref{x_y_1}, \eqref{x_y_2}, \eqref{overline_X}, \ \text{and } \ \eqref{SP1_pie_a}  \notag,\\
     \mathbf{SP1.2}:\quad&\min_{\mathbf{\overline{Y}}} \ f(\mathbf{\overline{Y}}), \label{SP1.2}\\
    \text{s.t.} \ \ &\eqref{AAl},\eqref{x_y_1}-\eqref{y_z}, \eqref{overline_Y}, \ \text{and} \ \eqref{SP1_pie_b}-\eqref{SP1_pie_e}  \notag,\\
     \mathbf{SP1.3}:\quad&\min_{\mathbf{\overline{Z}}} \ f(\mathbf{\overline{Z}}), \label{SP1.3}\\
    \text{s.t.} \ \ &\eqref{AAl},\eqref{AAf},\eqref{y_z}, \ \text{and} \ \eqref{overline_Z}.   \notag
\end{alignat}}

\par We present the solution of $\mathbf{SP1.1}$ by using the BSUMM method as an example for both $\mathbf{SP1.2}$ and $\mathbf{SP1.3}$. The main steps are illustrated as follows.

\par \textbf{First}, variable $\overline{\mathbf{X}}$ is split into $m$ smaller blocks based on the relationships between variables. For example, the offloading decisions of ISDs served by the same UAV can be considered to be a variable block. Therefore, subproblem $\mathbf{SP1.1}$ can be expressed as follows:

\vspace{-0.8em}
{\small
\begin{alignat}{1}
\label{SP1.1_pie}
    \mathbf{SP1.1^{\prime}}: \quad &\min\limits_{\mathbf{\overline{X}}} f(\boldsymbol{x}_1,\boldsymbol{x}_2,\ldots,\boldsymbol{x}_m),  \\
        \text{s.t.}\ \ &\boldsymbol{x}_i \in \mathcal{X}_i, \ i = 1,2,\ldots,m,\notag\\
        &\eqref{AAl},\eqref{x_y_1}, \eqref{x_y_2}, \eqref{overline_X} , \ \text{and} \ \eqref{SP1_pie_a}  \notag,
\end{alignat}}

\noindent where $\overline{\mathbf{X}}=(\boldsymbol{x}_1,\ldots,\boldsymbol{x}_m)$,  $\mathcal{X}=\mathcal{X}_1 \times \cdots \times \mathcal{X}_m \subseteq \mathbb{R}^l$, each $\mathcal{X}_i \subseteq \mathbb{R}^{l_i}$  is a closed convex set, and $l_i$ represents the dimension of the $i$-th block vector $\boldsymbol{x}_i$.

\par \textbf{Second}, at each iteration, one variable block $\boldsymbol{x}_i$ is selected to optimize based on the cyclic rule \cite{Razaviyayn2013}, i.e., $i$ is in the order of $\{1,2,\ldots,m,1,2,\ldots,m,\ldots\}$. Moreover, when optimizing variable $\boldsymbol{x}_i$ at the $r$-th iteration, an upper bound function $u_i(\cdot)$ is constructed to transform the problem $\mathbf{SP1.1}$ into a more tractable and faster-converge problem by adding a quadratic penalization as follows:
\begin{sequation}
\label{upper_bound_function}
u_i(\boldsymbol{x}_i,\boldsymbol{x}^{r-1})=f(\boldsymbol{x}_i,\boldsymbol{x}_{-i}^{r-1})+\frac{\sigma_0}{2}||\boldsymbol{x}_i-\boldsymbol{x}_i^{r-1}||^2,
\end{sequation}

\noindent where $u_i(\boldsymbol{x}_i,\boldsymbol{x}^{r-1})$ denotes an approximation function of $f(\boldsymbol{x}_i,\boldsymbol{x}^{r-1}_{-i})$ for each block $i$ at a given feasible point $\boldsymbol{x}^{r-1}\in\mathcal{X}$, $\boldsymbol{x}_{-i}^{r-1}:=(\boldsymbol{x}_1^{r-1},\ldots,\boldsymbol{x}_{i-1}^{r-1},\boldsymbol{x}_{i+1}^{r-1},\ldots,\boldsymbol{x}_n^{r-1})$ represents the blocks except for block $\boldsymbol{x}_i^{r-1}$, and $\frac{\sigma_0}{2}||\boldsymbol{x}_i-\boldsymbol{x}_i^{r-1}||^2 (\sigma_0 > 0)$ is the proximal term.

\par \textbf{Third}, by replacing the objective function (\ref{SP1.1_pie}) with the upper bound function  (\ref{upper_bound_function}), $\mathbf{SP1.1}$ can be reformulated as
\begin{sequation}
    \begin{aligned}
\mathbf{{SP1.1}^{\prime\prime}}: \quad &\min\limits_{\boldsymbol{x}_i} u_i(\boldsymbol{x}_i,\boldsymbol{x}^{r-1}), \\
       \text{s.t.} \ \ &h_j(\boldsymbol{x})=0,j\in\{ 1,2,\ldots,l_{eq} \},\\
       &g_{k}(\boldsymbol{x})\leq 0, k\in\{ 1,2,\ldots,l_{neq} \},
    \end{aligned}
\end{sequation}

\noindent where $h(\cdot)$ represents the linear equality constraint \eqref{overline_X}, and $g(\cdot)$ means the linear inequality constraints (\ref{x_y_1}), (\ref{x_y_2}), \eqref{AAl}, and (\ref{SP1_pie_a}). Moreover, $l_{eq}$ and $l_{neq}$ denote the numbers of linear equality and inequality constraints, respectively. To deal with the linear coupling constraints and better facilitate an application of BSUMM algorithm, we incorporate these constraints into the upper bound function $u_i(\cdot)$ by using the alternating direction method of multipliers algorithm~\cite{Hong2016}, and then obtain the upper bound augmented Lagrangian function $\tilde{L}_i(\cdot)$. Specifically, defining $\boldsymbol{\mu}:=\{\mu_1,\ldots,\mu_{l_{eq}}\}$ and $\boldsymbol{\lambda}:=\{\lambda_1,\ldots,\lambda_{l_{neq}}\}$ as the Lagrange multipliers corresponding to the linear equality and inequality constraints, respectively, $\mathbf{SP1.1''}$ can be rewritten as follows:
\begin{sequation}    
\label{problem_Lag}
    \begin{aligned}        
    &\mathbf{SP1.1^{\prime\prime\prime}}: \quad \min\limits_{\boldsymbol{x}_i}\tilde{L}_i(\boldsymbol{x}_i,\boldsymbol{x}^{r-1};\boldsymbol{\mu}^{r-1},\boldsymbol{\lambda}^{r-1})\\
    =\ &\min\limits_{\boldsymbol{x}_i}u_i(\boldsymbol{x}_i,\boldsymbol{x}^{r-1})+\sum\limits_{j=1}^{l_{eq}}[\mu_j^{r-1} h_j(\boldsymbol{x}_i)+\frac{\sigma}{2}h_j^2(\boldsymbol{x}_i)] \\
    +\ &\frac{1}{2\sigma}\sum\limits_{k=1}^{l_{neq}} \{ [ \max(0,\lambda_k^{r-1}+\sigma g_k(\boldsymbol{x}_i)) ]^2- (\lambda_k^{r-1})^2 \},
    \end{aligned}
\end{sequation}

\noindent where $\sigma$ is a penalty parameter. Moreover, the iterative process for optimizing $\tilde{L}_i$ involves updating the variable block $\boldsymbol{x}_i$ as well as the Lagrange multipliers $\boldsymbol{\mu}$ and $\boldsymbol{\lambda}$, which are given as follows:

\vspace{-0.8em}
{\small
\begin{subequations}
\begin{alignat}{1}
    &\quad\quad\boldsymbol{x}_i^r\in \arg \min\limits_{\boldsymbol{x}_i\in\overline{\mathbf{X}}}    \tilde{L}_i(\boldsymbol{x}_i,\boldsymbol{x}^{r-1};\boldsymbol{\mu}^{r-1},\boldsymbol{\lambda}^{r-1}),  \\
    &\quad\quad\boldsymbol{x}_{-i}^r=\boldsymbol{x}_{-i}^{r-1},   \\
    &\quad\quad\mu_j^{r}=\mu_j^{r-1}+\sigma h_j(\boldsymbol{x}^r), \ j = 1,2,\ldots,l_{eq}, \label{mu} \\
    &\quad\quad\lambda_k^{r}=\lambda_k^{r-1}+\sigma g_k(\boldsymbol{x}^r), \ k = 1,2,\ldots,l_{neq}.  \label{lambda}
\end{alignat}
\end{subequations}}

 \par To ensure that the algorithm can be convergent and does not violate the constraints of the original problem, the conditions need to be guaranteed, which is as follows \cite{Boyd2011}:

\vspace{-0.8em}
{\small
\begin{subequations}
    \begin{alignat}{1}
        &\Omega^r_1 = \left|\tilde{L}_i^{(r)}-\tilde{L}_i^{(r+1)}\right| \leq \epsilon_1, \label{Omega_1}\\
        &\Omega^r_2 = \big\{ \sum\limits_{j=1}^{l_{eq}}h_j^2(\boldsymbol{x}^r)+\sum\limits_{k=1}^{l_{neq}}\big[ \max\big(g_k(\boldsymbol{x}^r),-\lambda_k^r/\sigma\big) \big]^2\big\} ^{0.5} \leq \epsilon_2,\label{Omega_2}
    \end{alignat}
\end{subequations}}

\noindent where $\epsilon_1$ and $\epsilon_2$ are the  acceptable convergence gaps.

\par \textbf{Finally}, when the convergence conditions are met, the optimal solution of $\boldsymbol{x}^r$ can be obtained. However, $\boldsymbol{x}^r$ is a continuous variable within the closed interval of $[0,1]$, while the decision of computation offloading is a binary variable. Therefore, the threshold rounding technique \cite{Elbassioni2021} is applied to transform the relaxed $\boldsymbol{x}^r$ into binary variables. Specifically, each element $x^* \in \boldsymbol{x}^r$ is transformed as follows:
\begin{sequation}
    x^* \ = \ \left\{
    \begin{aligned}
        &1, \ \text{if} \ x^* \geq \delta,\\
        &0, \ \text{otherwise},
    \end{aligned}
    \right.  
\end{sequation}

\noindent where $\delta \in (0,1)$ is a positive rounding threshold. Therefore, the optimal decisions of computation offloading for ISDs can be obtained as $\mathbf{X}^*$. At this point, $\mathbf{SP1.1}$ is completely solved. Similarly, subproblems $\mathbf{SP1.2}$ and $\mathbf{SP1.3}$ can be solved by using the similar method. As a result, the optimal solutions of computation offloading and service caching can be obtained as $\mathbf{X}^*$, $\mathbf{Y}^*$, and $\mathbf{Z}^*$. 

\par However, there is an integrality gap to be concerned because of the rounding process from continuous variables to discrete variables \cite{Zhang2017}. Taking $\mathbf{Z}^*$ as an example, we use $\Delta_1$, $\Delta_2$, and $\Delta_3$ to denote the maximum violation degrees of constraints \eqref{AAf}, \eqref{AAl}, and \eqref{y_z}, respectively. As described in \cite{Feige2016}, we denote the integrality gap  as follows:
\begin{sequation}
    \label{integ_gap}
    \Xi = \min_{\overline{\mathbf{Z}}} \tilde{L}_i/(\tilde{L}_i+\xi \Delta), \forall i,
\end{sequation}

\noindent where $\Delta=\Delta_1+\Delta_2+\Delta_3$, and $\xi$ is the weight of $\Delta$. Moreover, the optimal solutions are accepted when $\Xi=1$, which is proved by Theorem \ref{integrality_gap}.

\begin{theorem}
    \label{integrality_gap}
    There is no violation of constraints when $\Xi=1$.
\end{theorem}

\begin{proof}
    By implementing the maximum violation, the constraints \eqref{AAf}, \eqref{AAl}, and \eqref{y_z} can be re-expressed as follows:

    \vspace{-0.8em}
    {\small
    \begin{subequations}
        \begin{alignat}{1}
        \sum\limits_{k\in\mathcal{K}}z_k^u[n]&\leq C_u^{\max}+\Delta_1,\\
        t_k[n]&<t_k^{\max}+\Delta_2,\\
        y_k^u[n]&\leq z_k^u[n]+\Delta_3,
    \end{alignat}
    \end{subequations}}
    
    \noindent where $\Delta_1$, $\Delta_2$, and $\Delta_3$ are obtained as follows:

    \vspace{-0.8em}
    {\small
    \begin{subequations}
        \begin{alignat}{1}
        \Delta_1&=\max\{\sum\limits_{k\in\mathcal{K}}z_k^u[n]-C_u^{\max},0\},\\
        \Delta_2&=\max\{t_k[n]-t_k^{\max},0\},\\
        \Delta_3&=\max\{y_k^u[n]-z_k^u[n],0\}.
    \end{alignat}
    \end{subequations}}
    
    \noindent $\Delta=\Delta_1+\Delta_2+\Delta_3=0$ indicates that no constraint is violated. Moreover, the solution for $\tilde{L}_i$ is derived by relaxing the variables, and the solution for $\tilde{L}_i + \xi \Delta$ is obtained by rounding the relaxed variables.  In addition, the optimal rounding can be achieved when $\Xi = 1$, and $\Delta = 0$ at this point. Therefore, no constraint is violated when $\Xi=1$.
\end{proof}

\par The joint optimization method of commutation offloading and service caching is summarized in Algorithm \ref{al_offloading_caching}. Specifically, a variable block $\boldsymbol{x}_i$ is selected based on the cyclic rule (line 3). Then, the variable block is optimized and updated after obtaining the upper bound augmented Lagrangian function (lines 4 to 5). Moreover, the variables of $\boldsymbol{y}_i$ and $\boldsymbol{z}_i$ are optimized by solving the problems of (\ref{SP1.2}) and (\ref{SP1.3}) (line 6). In addition, the Lagrange multipliers are updated based on Eqs. \eqref{mu} and \eqref{lambda} (line 8). This process is repeated until the convergence condition is satisfied (lines 9 and 10). Finally, the optimal decisions of computation offloading $\mathbf{O}^* = \{ \mathbf{X}^*, \mathbf{Y}^* \}$ and service caching $\mathbf{Z}^*$ can be obtained by using the rounding technique (lines 11 and 12).
\vspace{-1em}
\begin{algorithm}
\label{al_offloading_caching}
    \caption{Joint Commutation Offloading and Service Caching}
    \KwIn {$\hat{\mathbf{\Theta}}$, $\hat{\mathbf{F}}$, $\hat{\mathbf{Q}}$.}
    \textbf{Initialization:} $\sigma_0$, $\bold{\mu}^{(0)}$, $\bold{\lambda}^{(0)}$, $r = 0$, $\epsilon_1$, $\epsilon_2$, $(\overline{\mathbf{X}}^{(0)},\overline{\mathbf{Y}}^{(0)},\overline{\mathbf{Z}}^{(0)})$\;
    \Repeat{\text{satisfy}\ $\eqref{Omega_1}$ \ \text{and} \ $\eqref{Omega_2}$}
    {
        Select one variable block $\boldsymbol{x}_i$ to be updated\;
        Obtain function $\tilde{L}_i(\boldsymbol{x}_i,\boldsymbol{x}^{r}_{-i};\boldsymbol{\mu}^r,\boldsymbol{\lambda}^r)$\;
        Solve problem $\boldsymbol{x}_i^{r+1} \in \min\limits_{\boldsymbol{x}_i\in\overline{\mathbf{X}}} \tilde{L}_i(\boldsymbol{x}_i,\boldsymbol{x}^{r}_{-i};\boldsymbol{\mu}^r,\boldsymbol{\lambda}^r)$\;
        Set $\boldsymbol{x}_{-i}^{r+1} = \boldsymbol{x}_{-i}^{r}$\;
        Solve (\ref{SP1.2}) and (\ref{SP1.3}) to obtain $\boldsymbol{y}_i^{r+1}$ and $\boldsymbol{z}_i^{r+1}$\;
        Update $\boldsymbol{\mu}^{r}$ and $\boldsymbol{\lambda}^{r}$ based on \eqref{mu} and \eqref{lambda}\;
        Update $r = r + 1$\;
    }
    Generate binary solutions $\bold{X}^*$, $\bold{Y}^*$, and $\bold{Z}^*$ for $\boldsymbol{x}^{r+1}$, $\boldsymbol{y}^{r+1}$, and $\boldsymbol{z}^{r+1}$ by using the rounding technique\;
    \KwOut{$\mathbf{O}^*=\{\mathbf{X}^*,\mathbf{Y}^*\}$, $\mathbf{Z}^*$}
\end{algorithm}

\vspace{-2em}
\subsection{Communication and Computation Resource Allocation}
\par Based on the obtained decisions of computation offloading $\mathbf{O}^*$ and service caching $\mathbf{Z}^*$, along with the given UAV trajectory $\hat{\mathbf{Q}}$, the original optimization problem $\mathbf{P}$ can be transformed into subproblem $\mathbf{SP2}$ to determine the communication resource allocation $\mathbf{\Theta}$ and computation resource allocation $\mathbf{F}$, which is formulated as follows:

\vspace{-0.8em}
{\small
\begin{alignat}{1}
    \mathbf{SP2}: \ &\min_\mathbf{\Theta,F} \ \sum\limits_{k\in\mathcal{K}}t_k[n]\\
    &=\min_\mathbf{\Theta,F} \ \sum\limits_{k\in\mathcal{K}} d_k\big( \frac{o_k^u[n]+\sum\limits_{\substack{v\in\mathcal{U},v\neq u}}o_k^v[n]+o_k^M[n]}{\theta_{k,u}[n]B_{f}\log_2(1+P_{k}G_{k,u}[n]/\sigma^2)}\notag\\
    &+\sum\limits_{\substack{v\in\mathcal{U}, v\neq u}}\frac{o_k^v[n]}{\theta_{u,v}[n]B_{c}\log_2(1+P_{u}G_{u,v}[n]/\sigma^2)}\notag\\
    &+\frac{o_k^M[n]}{\theta_{u,M}[n]B_{b}\log_2(1+P_{u}G_{u,M}[n]/\sigma^2)}+c_k(o_k^u[n]/f_k^u[n]\notag\\
    &+\sum\limits_{\substack{v\in\mathcal{U},v\neq u}}o_k^v[n]/f_k^v[n]+o_k^M[n]/f_k^M[n])\big), \label{SP2} \\
    \text{s.t.} \ &\eqref{AAa}-\eqref{AAe}, \ \text{and} \ \eqref{AAl}\notag.
\end{alignat}}


\par The subproblem $\mathbf{SP2}$ is a convex optimization problem that can be efficiently solved by using the tools such as CVX \cite{Grant2014}, and the method of joint communication and computation resource allocation is summarized in Algorithm \ref{algorithm_resource_allocation}. Specifically, in the $r$-th iteration, the optimal decisions of communication resource allocation $\mathbf{\Theta}^{r+1}$ and computation resource allocation $\mathbf{F}^{r+1}$ are obtained by solving the subproblem $\mathbf{SP2}$ (line 3). Then, repeat this process until the difference of the objective value falls below a given threshold $\epsilon$ between two successive iterations (lines 4 and 5).

\vspace{-0.8em}
\begin{algorithm}
    \label{algorithm_resource_allocation}
    \caption{Joint Communication and Computation Resource Allocation}
    \KwIn {$\mathbf{O}^*=\{\mathbf{X}^*,\mathbf{Y}^*\}$, $\mathbf{Z}^*$, $\hat{\mathbf{Q}}$, $\epsilon$}
    \textbf{Initialization:} $\epsilon$, $\mathbf{\Theta}^{(0)}$,$\mathbf{F}^{(0)}$,$r = 0$\;
    \Repeat{$|G^{r+1}-G^{r}|<\epsilon$}
    {
        Solve problem $\mathbf{SP2}$ to obtain the decisions $\mathbf{\Theta}^{r+1}$, $\mathbf{F}^{r+1}$, and the objective value  $G^{r+1}$\;
        Update $r = r + 1$\; 
    }
    \KwOut {$\mathbf{\Theta}^{*}$ and $\mathbf{F}^{*}$}
\end{algorithm}

\vspace{-1.8em}
\subsection{UAV Trajectory Control}

\par Based on the obtained computing offloading $\mathbf{O}^*$, service caching $\mathbf{Z}^*$, communication resource allocation $\mathbf{\Theta}^*$, and computation resource allocation $\mathbf{F}^*$, problem $\mathbf{P}$ can be transformed as the UAV trajectory control subproblem as:

\vspace{-0.8em}
{\small
\begin{alignat}{1}
    \mathbf{SP3}: \quad &\min\limits_{\mathbf{Q}} \ \sum\limits_{k\in\mathcal{K}} t_k[n],\label{SP3}\\
    \text{s.t.} \ \ &\eqref{AAi}-\eqref{AAl}\notag.
\end{alignat}}

\par It can be inferred that problem $\mathbf{SP3}$ is non-convex because the objective function, as well as the left-hand-side of constraints (\ref{AAj}) and (\ref{AAl}) are neither convex nor concave with respect to UAV trajectory $\boldsymbol{q}[n]$. Therefore, we convert the non-convex problem into a convex problem by the following steps.

%
%
\textbf{First}, to deal with the non-convex objective function \eqref{SP3}, we introduce the upper bound of the objective function, which is as follows:

\vspace{-0.8em}
{\small
\begin{alignat}{1}
    &t_k[n] \leq d_k \bigg( \frac{o_k^u[n]+\sum\limits_{\substack{v\in\mathcal{U}, v\neq u}}o_k^v[n]+o_k^M[n]}{\hat{R_{k,u}}[n]}+\sum\limits_{\substack{v\in\mathcal{U}, v\neq u}}\frac{o_k^v[n]}{\hat{R}_{u,v}[n]}+\notag\\
    &\frac{o_k^M[n]}{\hat{R}_{u,M}[n]}\bigg)+c_k\bigg(\frac{o_k^u[n]}{f_k^u[n]}+\sum\limits_{\substack{v\in\mathcal{U}, v\neq u}}\frac{o_k^v[n]}{f_k^v[n]}+\frac{o_k^M[n]}{f_k^M[n]}\bigg)g = \hat{t}_k[n],
\end{alignat}
}

\noindent where $\hat{R}_{k,u}[n]$, $\hat{R}_{u,v}[n]$, and $\hat{R}_{u,M}[n]$ represent the lower bounds of $R_{k,u}[n]$, $R_{u,v}[n]$, and $R_{u,M}[n]$, respectively, which are obtained by Theorem \ref{lower_bound_of_R}. 

\begin{theorem}
\label{lower_bound_of_R}
    Given the local points $\boldsymbol{q}^r_u$ and $\boldsymbol{q}^r_v (v\neq u)$ at the $r$-th iteration, $R_{k,u}[n]$, $R_{u,v}[n]$, and $R_{u,M}[n]$ are lower bounded by the constraints as follows:

\vspace{-0.8em}
{\small
    \begin{subequations}
        \begin{alignat}{1}
        R_{k,u}[n]\geq &\nabla R_{k,u}^r[n](||\boldsymbol{q}_u[n]-\boldsymbol{q}_k||^2-||\boldsymbol{q}_u^r[n]-\boldsymbol{q}_k||^2)\notag\\
        +&R_{k,u}^r[n]=\hat{R}_{k,u}[n],\label{R_k_u_lower_bound}\\
        R_{u,v}[n]\geq &\nabla R_{u,v}^r[n](||\boldsymbol{q}_u[n]-\boldsymbol{q}_v[n]||^2
        -||\boldsymbol{q}_u^r[n]-\boldsymbol{q}_v^r[n]||^2)\notag\\
        +&R_{u,v}^r[n]=\hat{R}_{u,v}[n],\label{R_u_v_lower_bound}\\
        R_{u,M}[n]\geq &\nabla R_{u,M}^r[n](||\boldsymbol{q}_u[n]-\boldsymbol{q}_M||^2-||\boldsymbol{q}_u^r[n]-\boldsymbol{q}_M||^2)\notag\\
        &+R_{u,M}^r[n]=\hat{R}_{u,M}[n],\label{R_u_M_lower_bound}
    \end{alignat}
    \end{subequations}}
    
    \noindent where $R_{k,u}^r[n]$, $R_{u,v}^r[n]$, and $R_{u,M}^r[n]$ are the data transmission rates at the $r$-th iteration for ISD-UAV link, UAV-UAV link, and UAV-MBS link, respectively. Moreover, $\nabla R_{k,u}^r[n]$ denotes the first-order derivative of $R_{k,u}^r[n]$ with respect to $||\boldsymbol{q}_u[n]-\boldsymbol{q}_k||^2$ at the $r$-th iteration, which can be given as
    \vspace{-0.8em}
    
{\small
\begin{align}
&R_{k,u}^r[n]=\theta_{k,u}[n]B_{f}\log_2\Big(1+ \frac{\gamma_k}{||\boldsymbol{q}_u^r[n]-\boldsymbol{q}_k||^2+H_u^2}\Big),\\
&\nabla R_{k,u}^r[n]=\frac{-\theta_{k,u}[n]B_{f} \gamma_k \log_2e}{||\boldsymbol{q}_u^r[n]-\boldsymbol{q}_k||^2+H_u^2}\times\frac{1}{||\boldsymbol{q}_u^r[n]-\boldsymbol{q}_k||^2+H_u^2+\gamma_k}.
\end{align}}

\noindent Besides, $\nabla R_{u,v}^r[n]$ and $\nabla R_{u,M}^r[n]$ denote the first-order derivatives of $R_{u,v}^r[n]$ and $R_{u,M}^r[n]$, respectively, which are calculated as follows:

{\small
\begin{align}
&R_{u,v}^r[n]=\theta_{u,v}[n]B_{c}\log_2\Big(1+ \frac{\gamma_u}{||\boldsymbol{q}_u^r[n]-\boldsymbol{q}_v^r[n]||^2}\Big),\\
&\nabla R_{u,v}^r[n]=\frac{-\theta_{u,v}[n]B_{c}\gamma_u \log_2e}{ (||\boldsymbol{q}_u^r[n]-\boldsymbol{q}_v^r[n]||^2 }\times \frac{1}{ ||\boldsymbol{q}_u^r[n]-\boldsymbol{q}_v^r[n]||^2+\gamma_u},
\end{align}
 }

{\small
        \begin{align}
        &R_{u,M}^r[n]=\theta_{u,M}[n]B_{b}\log_2\big(1+ \frac{\gamma_u}{||\boldsymbol{q}_u^r[n]-\boldsymbol{q}_M||^2+(H_u-H_M)^2}\big),\\
        &\nabla R_{u,M}^r[n]=\frac{-\theta_{u,M}[n]B_{b}\gamma_u\log_2e}{||\boldsymbol{q}_u^r[n]-\boldsymbol{q}_M||^2+(H_u-H_M)^2}\notag \\&\times\frac{1}{||\boldsymbol{q}_u^r[n]-\boldsymbol{q}_M||^2+(H_u-H_M)^2+\gamma_u},
    \end{align}}
    
\noindent where $\gamma_k = P_{k}\beta_0/\sigma^2$.
\end{theorem}
\begin{proof}
    Theorem \ref{lower_bound_of_R} is proved by considering $R_{k,u}[n]$ as one of the three cases. Specifically, $R_{k,u}[n]$ is convex with respect to $||\boldsymbol{q}_u[n]-\boldsymbol{q}_k||^2$. Therefore, the lower bound \eqref{R_k_u_lower_bound} of $R_{k,u}[n]$ can be obtained by the first order Taylor expansion of $||\boldsymbol{q}_u[n]-\boldsymbol{q}_k||^2$ \cite{Boyd2014}. For 
    $R_{u,v}[n]$ and $R_{u,M}[n]$, \eqref{R_u_v_lower_bound} and \eqref{R_u_M_lower_bound} can be achieved in a similar way due to their same functional form to $R_{k,u}[n]$.
\end{proof}

\par \textbf{Second}, by taking the first order Taylor expansion at the given point $\boldsymbol{q}_u^r[n]$ and $\boldsymbol{q}_v^r[n] (u\neq v)$ at the $r$-th iteration, the non-convex constraint \eqref{AAj}, is relaxed as follows:

\vspace{-0.8em}
{\small
    \begin{alignat}{1}
    \label{AAj_gai}
        ||\boldsymbol{q}_u[n]-\boldsymbol{q}_v[n]||^2 &\geq 2(\boldsymbol{q}_u^r[n]-\boldsymbol{q}_v^r[n])^T(\boldsymbol{q}_u[n]-\boldsymbol{q}_v[n])\notag\\
        &-||\boldsymbol{q}_u^r[n]-\boldsymbol{q}_v^r[n]||^2\geq (D^{\min})^2.
    \end{alignat}}

\par \textbf{Finally}, according to Theorem \ref{lower_bound_of_R}, the non-convex constraint \eqref{AAl} can be bounded as

\vspace{-0.8em}
{\small
\begin{alignat}{1}
\label{hat_t_k}
    \hat{t}_k[n] \leq t_k^{\max}.
\end{alignat}}


\par Based on the abovementioned steps, the subproblem $\mathbf{SP3}$ can be rewritten as

\vspace{-0.8em}
{\small
\begin{alignat}{1}
    \mathbf{SP3^{\prime}}:\quad &\min\limits_{\mathbf{Q}} \ \sum\limits_{k\in\mathcal{K}} \hat{t}_k[n],\label{SP3_pie}\\
        \text{s.t.} \ \ &(\ref{AAi}),(\ref{AAk}),\eqref{AAj_gai}, \ \text{and} \ \eqref{hat_t_k}.\notag
\end{alignat}}

\par Problem $\mathbf{SP3^{\prime}}$ is obviously a convex problem. Therefore, we adopt the standard convex optimization tools, such as the CVX solver, to iteratively solve the problem. We summarize the optimization process of problem $\mathbf{SP3^{\prime}}$ in Algorithm \ref{al_UAV_traj}. First, we solve the subproblem $\mathbf{SP3^{\prime}}$ in the $r$-th iteration and obtain the optimal UAV location $Q$ as the local point for the next iteration (line 3). Then, this process is repeated until the difference of the objective function value falls below a given threshold $\epsilon$ between two successive iterations (lines 4 and 5).

\begin{algorithm}
    \caption{UAV Trajectory Control}
    \label{al_UAV_traj}
    \KwIn {$\mathbf{O}^*$, $\mathbf{Z}^*$, $\mathbf{\Theta}^*$, $\mathbf{F}^*$, $\epsilon$}
    \textbf{Initialization:} $\mathbf{Q}^{(0)}$, $r = 0$\;
    \Repeat{$|G^{r+1}-G^{r}|<\epsilon$}
    {
        Solve the convex problem $\mathbf{SP3^{\prime}}$ to obtain the current optimal variable $\mathbf{Q}^{r+1}$, and the objective value $G^{r+1}$\;
        Update $r = r + 1$\;
    }
    \KwOut{$\mathbf{Q}^{*}$}
\end{algorithm}
\vspace{-1em}
\subsection{Main Steps of J$\text{C}^5$A}

\par The main steps of the proposed J$\text{C}^5$A are outlined in Algorithm \ref{C6}. Specifically, the decisions of computation offloading and service caching are obtained through Algorithm \ref{al_offloading_caching} (line 3). Moreover, the communication and computation resource allocations are determined by Algorithm \ref{algorithm_resource_allocation} (line 4). Then, the UAV trajectory control is determined via Algorithm \ref{al_UAV_traj} (line 5). In addition, the total delay is computated (line 6). Finally, update the decisions and iterate the above steps until the algorithm converges (lines 7 to 9).

\subsection{Analysis of J$\text{C}^5$A}

\par The related convergence and complexity analysis of  J$\text{C}^5$A are presented as follows.

\subsubsection{Convergence Analysis} 

\par The convergence analysis of the proposed $\text{JC}^5$A is given by Theorem \ref{th_P_is_converged}.

\begin{theorem}
    \label{th_P_is_converged}
    The $\text{JC}^5$A  given in Algorithm \ref{C6} can be converged within a finite number of iterations.
\end{theorem}
\begin{proof}
    Algorithm \ref{C6} consists of two levels of iterations. Specifically, the outer level aims to optimize the strategies by iteratively executing Algorithms 1, 2, and 3, while the inner level aims to optimize the corresponding strategies returned by Algorithms 1, 2 and 3 in each iteration of the outer level. Moreover, the iterative structures of the inner and outer levels are similar. Namely, they both optimize the objective function or the upper bound function. Therefore, we use the outer level as an example for demonstration, which is detailed as follows.
    
    \par \textbf{First}, at the $(r+1)$-th iteration, the optimal $\mathbf{O}^{r+1}$ and  $\mathbf{Z}^{r+1}$ with given  $\mathbf{\Theta}^r$, $\mathbf{F}^r$, and  $\mathbf{Q}^r$ can be achieved by solving problem $\mathbf{SP1}$. Therefore, the objective value can be bounded as follows:

    \vspace{-0.8em}
{\small
    \begin{alignat}{1}
        \label{convergence_proof_1}
t(\mathbf{O}^{r},\mathbf{Z}^r,\mathbf{\Theta}^r,\mathbf{F}^r,\mathbf{Q}^r)&\geq t_{\mathbf{O},\mathbf{Z}}^{r,ub}(\mathbf{O}^{r+1},\mathbf{Z}^{r+1},\mathbf{\Theta}^r,\mathbf{F}^r,\mathbf{Q}^r)\notag\\
        &\geq t_{\mathbf{O},\mathbf{Z}}^{r}(\mathbf{O}^{r+1},\mathbf{Z}^{r+1},\mathbf{\Theta}^r,\mathbf{F}^r,\mathbf{Q}^r),
    \end{alignat}}
    \noindent where $t(\mathbf{O}^{r},\mathbf{Z}^r,\mathbf{\Theta}^r,\mathbf{F}^r,\mathbf{Q}^r)$ denotes the objective value at the $r$-th iteration, 
    $t_{\mathbf{O},\mathbf{Z}}^{r}(\cdot)$ represents the objective value of solving computation offloading and service caching, and $t_{\mathbf{O},\mathbf{Z}}^{r,ub}(\cdot)$ denotes the upper bound function of $t_{\mathbf{O},\mathbf{Z}}^{r}(\cdot)$. 
    
   \par \textbf{Second}, with the optimized $\mathbf{O}^{r+1}$ and $\mathbf{Z}^{r+1}$, and given $\mathbf{Q}^r$, the optimal $\mathbf{\Theta}^{r+1}$ and $\mathbf{F}^{r+1}$ can be obtained in the ($r+1$)-th iteration as follows:

   \vspace{-0.8em}
{\small
    \begin{alignat}{1}
    \label{convergence_proof_2}
        &t_{\mathbf{O},\mathbf{Z}}^{r}(\mathbf{O}^{r+1},\mathbf{Z}^{r+1},\mathbf{\Theta}^r,\mathbf{F}^r,\mathbf{Q}^r)\notag  \\
        = \ &  t_{\mathbf{\Theta},\mathbf{F}}(\mathbf{O}^{r+1},\mathbf{Z}^{r+1},\mathbf{\Theta}^r,\mathbf{F}^r,\mathbf{Q}^r)\notag\\
        \geq \ & t_{\mathbf{\Theta},\mathbf{F}}(\mathbf{O}^{r+1},\mathbf{Z}^{r+1},\mathbf{\Theta}^{r+1},\mathbf{F}^{r+1},\mathbf{Q}^r),
    \end{alignat}}
    
    \noindent where $t_{\mathbf{\Theta},\mathbf{F}}(\cdot)$ denotes the objective value of solving communication and computation resource allocation.

    \par \textbf{Third}, with the optimized $\mathbf{O}^{r+1}$ and $\mathbf{Z}^{r+1}$, the optimal trajectory $\mathbf{Q}^{r+1}$ can be obtained by solving the upper bound of the objective function. Consequently, the objective value can be bounded as

    \vspace{-0.8em}
{\small
    \begin{alignat}{1}
    \label{convergence_proof_3}
        &t_{\mathbf{\Theta},\mathbf{F}}(\mathbf{O}^{r+1},\mathbf{Z}^{r+1},\mathbf{\Theta}^{r+1},\mathbf{F}^{r+1},\mathbf{Q}^r) \notag\\
        = \ & t_{\mathbf{Q}}^{r,ub}(\mathbf{O}^{r+1},\mathbf{Z}^{r+1},\mathbf{\Theta}^{r+1},\mathbf{F}^{r+1},\mathbf{Q}^r)\notag\\
        \geq \ & t_{\mathbf{Q}}^{r,ub}(\mathbf{O}^{r+1},\mathbf{Z}^{r+1},\mathbf{\Theta}^{r+1},\mathbf{F}^{r+1},\mathbf{Q}^{r+1})\notag\\
        \geq \ & t(\mathbf{O}^{r+1},\mathbf{Z}^{r+1},\mathbf{\Theta}^{r+1},\mathbf{F}^{r+1},\mathbf{Q}^{r+1}),
    \end{alignat}}
    
    \noindent where $t_{\mathbf{Q}}^{r,ub}(\cdot)$ denotes the upper bound of the objective value of solving the UAV trajectory.
    
    \par \textbf{Finally}, combining \eqref{convergence_proof_1} to \eqref{convergence_proof_3}, we have
    \begin{sequation}  
    \label{eq_non_increasing}t(\mathbf{O}^{r},\mathbf{Z}^{r},\mathbf{\Theta}^{r},\mathbf{F}^{r},\mathbf{Q}^{r})\geq t(\mathbf{O}^{r+1},\mathbf{Z}^{r+1},\mathbf{\Theta}^{r+1},\mathbf{F}^{r+1},\mathbf{Q}^{r+1}).
    \end{sequation}
    
    \par  It can be deduced from \eqref{eq_non_increasing} that the objective value of $\mathbf{P}$ is non-increasing over each iterative. Therefore, the convergence of the proposed J$\text{C}^5$A is proved.
\end{proof}

\vspace{-1.5em}
\begin{algorithm}
\label{C6}
    \caption{J$\text{C}^5$A}
    \textbf{Initialization:} $\mathbf{O}^{(0)}$, $\mathbf{Z}^{(0)}$, $\mathbf{\Theta}^{(0)}$, $\mathbf{F}^{(0)}$,$\mathbf{Q}^{(0)}$,$t^{(0)}$,$r = 0$\;
    \Repeat{$|t^{(r)}-t^{(r-1)}|< \epsilon$}
    {
        Call Algorithm \ref{al_offloading_caching} to obtain $\mathbf{O}^{*}$ and $\mathbf{Z}^{*}$\;
        Call Algorithm \ref{algorithm_resource_allocation} to obtain $\mathbf{\Theta}^{*}$ and $\mathbf{F}^{*}$\;
        Call Algorithm \ref{al_UAV_traj} to obtain $\mathbf{Q}^{*}$\;
        Compute the total delay $t^{(r)}$\;
        Update $\mathbf{O}^{(r+1)} = \mathbf{O}^*$, $\mathbf{Z}^{(r+1)} = \mathbf{Z}^*$, $\mathbf{\Theta}^{(r+1)} = \mathbf{\Theta}^*$, $\mathbf{F}^{(r+1)} = \mathbf{F}^*$, and $\mathbf{Q}^{(r+1)} = \mathbf{Q}^*$\;
        Update $r = r + 1$\;
    }
    \Return{$\mathbf{O}^*$, $\mathbf{Z}^*$, $\mathbf{\Theta}^*$, $\mathbf{F}^*$, $\mathbf{Q}^*$}\;
\end{algorithm}

\subsubsection{Computational Complexity} 

\par The computational complexity of the proposed algorithm is given as Theorem $\ref{th_complexity}$. 
\begin{theorem}
    \label{th_complexity}
    The proposed algorithm has a polynomial worst-case complexity in each time slot, i.e., $\mathcal{O}(R_4(R_1N_1(\log(1/\epsilon)+1)+R_2N_2+R_3N_3))$, where $R_1$, $R_2$, $R_3$, and $R_4$ represent the numbers of iterations required for Algorithm 1, Algorithm 2, Algorithm 3, and Algorithm 4, respectively. Moreover, $\epsilon$ denotes the search accuracy. Besides, $N_1$, $N_2$, and $N_3$  are the numbers of variables in $\mathbf{SP1}$, $\mathbf{SP2}$ and $\mathbf{SP3}$, respectively.
\end{theorem}

\begin{proof}
    For the subproblem $\mathbf{SP1}$, it can be derived that BSUMM takes $\mathcal{O}(\log(1/\epsilon))$ iterations to find an $\epsilon$-optimal solution, which is known as the sub-linear convergence \cite{Razaviyayn2013}. Moreover, problem $\mathbf{SP1}$ has $N_1 = K+K(U+1)+KU$ decision variables, including $K$ variables of $\overline{\mathbf{X}}$, $K(U+1)$ variables of $\overline{\mathbf{Y}}$, and $KU$ variables of $\overline{\mathbf{Z}}$. Therefore, the computational complexity of $\mathbf{SP1}$ can be obtained as $C_1 = \mathcal{O}(\log(1/\epsilon)N_1+N_1)$. For the subproblem $\mathbf{SP2}$, let $N_2 = U^2+K(U+1)+2U$ denote the number of optimal variables, and the computational complexity can be easily given as $C_2 = \mathcal{O}(N_2)$. For the subproblem $\mathbf{SP3}$, the number of variables is $N_3 = 2U$, and then the computational complexity is $C_3 = \mathcal{O}(N_3)$. Finally, denoting the numbers of iterations required in solving these three subproblems and problem $\mathbf{P}$ as $R_1$, $R_2$, $R_3$, and $R_4$, respectively, the complexity of Algorithm 4 is $\mathcal{O}(R_4(R_1N_1(\log(1/\epsilon)+1)+R_2N_2+R_3N_3))$.
\end{proof}

\section{Simulation Results and Analysis}
\label{sec:Simulation Results And Analysis}

\par In this section, we conduct simulations to verify the effectiveness of the proposed J$\text{C}^5$A approach.

\subsection{Simulation Setup}

\par This section presents the scenarios, parameters, evaluation metrics, and comparison approaches.

\subsubsection{Scenarios}

\par  We consider a collaborative aerial MEC-assisted ICPS that consists of 30 random-distributed ISDs, an MBS, and 4 UAVs within the area of $1000 \times 1000$ $\text{m}^2$. Moreover, the system timeline is set as 50 s, which is divided into 50 time slots with equal length of 1 s.

\subsubsection{Parameters}

\par  The altitude of UAVs is set as $H$ = 100 m, and the initial positions of UAVs are set as $\boldsymbol{q}_1^I=[250,250]$, $\boldsymbol{q}_2^I=[750,250]$, $\boldsymbol{q}_3^I=[250,750]$, and $\boldsymbol{q}_4^I=[750,750]$. The simulation parameters are summarized in Table \ref{tab_simuParameter}.

\begin{table}[t]
\caption{Simulation Parameters}
\label{tab_simuParameter}
    \centering
    \begin{tabular}{c | c || c | c}
    \toprule
    Parameters & Values & Parameters & Values \\
    \midrule
    $\epsilon$ & $10^{-3}$ &   $\beta_0$ & $10^{-5}$ \\
    $B_f$ & $15$ MHz & $B_c$ & $10$ MHz \cite{Wang2022} \\
    $B_b$ & $5$ MHz  & $N$ & $50$ \cite{Zhou2017}\\ $\tau$ & 1s & $\sigma^2$ & $1.0\times 10^{-12}$ W \\
    $H$ & 100 m \cite{Zhang2020a} & $H_M$ & 25 m \\
    $F_k^{\max}$ & [0.5,1] GHz & $d_k$ & [0.5,3] Mb \\
    $c_k$ & [300,600] cycles/bit & $P_k$ & [10,20] dBm\\
    $F_u^{\max}$ & [15,20] GHz & $P_u$ & [20,23] dBm\\
    $C_u^{\max}$ & [5,10] & $\eta$ & 2.0 \\
    $F_M^{\max}$ & 20 GHz & $t_k$ & [0.5,1.0] s \\ 
    $V^{\max}$ & 50 m/s & $D^{\min}$ & 10 m \cite{Ji2020}\\
    \bottomrule
    \end{tabular}
\end{table}

\subsubsection{Evaluation Metrics}

\par  To evaluate the overall performance of the proposed J$\text{C}^5$A, we adopt the following indicators.  \textit{i)} Average completion delay (ACD) $(\sum\limits_{n\in\mathcal{N}}\sum\limits_{k\in\mathcal{K}}t_k[n]/|\mathcal{K}|)/N$, which indicates the average delay for completing the tasks of ISDs successfully. \textit{ii)} Average processing rate (APR) $\sum\limits_{n\in\mathcal{N}}\sum\limits_{k\in\mathcal{K}}d_k[n]c_k[n]/\sum\limits_{n\in\mathcal{N}}\sum\limits_{k\in\mathcal{K}}t_k[n]$, which indicates the average number of CPU cycles that are completed per time slot. \textit{iii)} Average service caching hit ratio (ASCHR) $(\sum\limits_{n\in\mathcal{N}}\sum\limits_{k\in\mathcal{K}} \sum\limits_{u\in\mathcal{U}} o_k^u[n]/\sum\limits_{u\in\mathcal{U}}C_u^{\max})/N$, which indicates the ratio of success cache hits of UAVs.

\subsubsection{Comparison Approaches}

\par  This work evaluates the J$\text{C}^5$A in comparison with the following approaches. \textit{i) Local computing (LC)}: All tasks are processed on ISDs locally. \textit{ii) All offloading (AO)}: The computation tasks are offloaded to UAVs or the MBS. \textit{iii) Static UAV (SU)}: The UAVs are deployed at stationary points without trajectory control. \textit{iv) Equal bandwidth and computing capacity (EBCC)~\cite{Zhou2021}}: The communication and computation resources are allocated evenly. \textit{v) Penalty successive convex approximation (P-SCA)}~\cite{Xu2022}: The decisions of computation offloading and service caching are optimized by the centralized method of P-SCA instead of BSUMM.  \textit{vi) Proximal policy optimization (PPO)-based joint optimization (PJO)~\cite{xu2024trajectory}}: All decisions of this work are jointly determined by the PPO algorithm. \textit{vii) Deep deterministic policy gradient (DDPG)-based joint optimization (DJO)~\cite{Miao2024}}: All decisions of this work are jointly determined by the DDPG algorithm.

\subsection{Evaluation Results}

\par In this section, we first evaluate the effectiveness of the proposed J$\text{C}^5$A with default parameters. Then, we investigate the impacts of different parameters on the performance of the proposed J$\text{C}^5$A and the comparison approaches.

\subsubsection{Effectiveness Evaluation}  Fig.~\ref{fig_SU} illustrates the convergence of the proposed J$\text{C}^5$A with different penalty parameters. We can observe that the figure that as the number of iterations increases, the average completion delay decreases and eventually converges, aligning with the analysis in Theorem \ref{th_P_is_converged}. Besides, a smaller parameter $\rho$ leads to a faster convergence speed. The reason is that a smaller penalty parameter results in a tighter approximation of the upper bound for the original objective function, thereby facilitating faster convergence. However, when $\rho = 0$, the convergence curve diminishes significantly because it becomes more complex to solve the original optimization problem directly. Therefore, the results in Fig.~\ref{fig_SU} demonstrate that the proposed J$\text{C}^5$A with a reasonable penalty parameter can reach a stable state within a finite number of iterations.

\begin{figure}[!hbt] 
\vspace{-1em}
    \centering
    \setlength{\abovecaptionskip}{0pt}%
    \setlength{\belowcaptionskip}{0pt}%
    \includegraphics[width =3in]{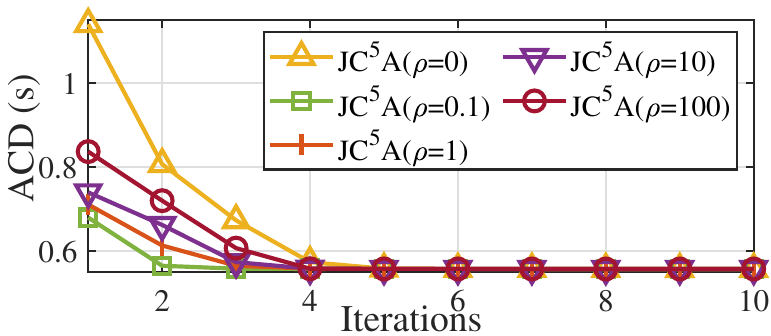}
    \caption{Convergence of J$\text{C}^5$A.}
    \label{fig_SU}
    \vspace{-0.9 em}
\end{figure}

\subsubsection{Impact of Parameters.} We evaluate the impact of UAV computation resources, the number ISDs, and UAV  service caching storage in this subsection.

\par \textbf{\textit{Impact of UAV Computation Resources.}} Figs. \ref{fig_1}(a), \ref{fig_1}(b), and \ref{fig_1}(c) show the impact of the average computation resource of UAVs on ACD, APR, and ASCHR, respectively, for the considered eight approaches. It can be observed from Fig. \ref{fig_1} that as the computing resource of UAVs increases, the proposed J$\text{C}^5$A consistently outperforms the other approaches in terms of ACD and APR, while showing a moderate performance in ASCHR among the eight approaches. This is due to the following reasons. First, the worst performance of LC is because it does not rely on the computing resources of UAVs, thereby missing the advantages of distributed processing and caching provided by the UAVs. Moreover, the entire UAV offloading strategy of AO, the static trajectory control of SU, the equal resource allocation method of EBCC, and the centralized control of P-SCA make these approaches highly sensitive to the resources of UAVs, resulting in less efficient resource utilization. Besides, the less efficiency of PCCT and DCCT is due to the long training periods, which is a result of the complex and hybrid action spaces of the problem. 

\par Although the proposed J$\text{C}^5$A achieves suboptimal performance in service caching hit ratio, it mitigates the additional costs associated with frequent service caching by exploiting the collaborative capabilities among the UAVs. As a result, the computing resources of UAVs can be fully utilized, leading to the decreased service delay and increased processing rate. In summary, the simulation results show that the proposed J$\text{C}^5$A is able to achieve superior performances in terms of ACD and APR by effectively utilizing the available computation resource of UAVs. This makes it particularly  suitable for our considered delay-sensitive and computation-intensive aerial MEC-assisted ICPS system. 

\begin{figure*}[!hbt] 
	\centering
	\subfigure[Average completion delay]
	{
		\begin{minipage}[t]{0.31\linewidth}
			\raggedleft
			\includegraphics[width=2.3in]{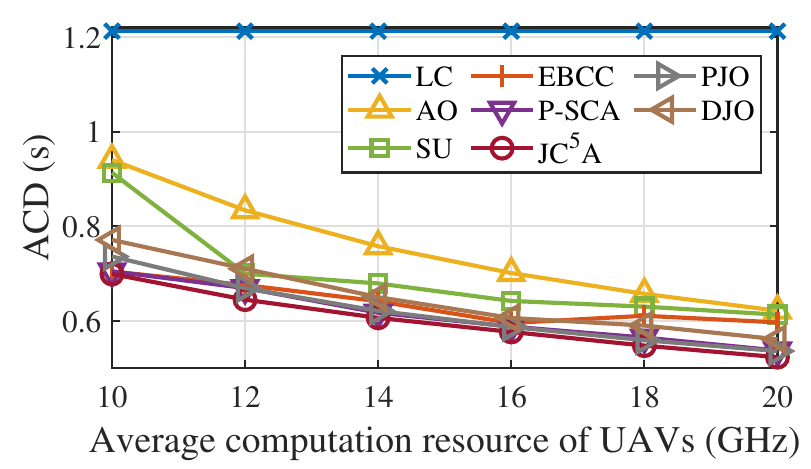}
		\end{minipage}
	}
	\subfigure[Average processing rate]
	{
		\begin{minipage}[t]{0.31\linewidth}
			\centering
			\includegraphics[width=2.3in]{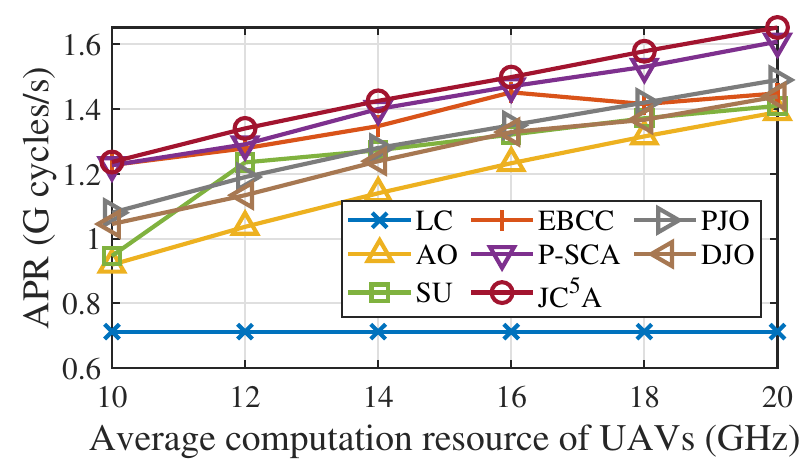}	
		\end{minipage}
	}
	\subfigure[Average service caching hit ratio]
	{
		\begin{minipage}[t]{0.31\linewidth}
			\centering
			\includegraphics[width=2.3in]{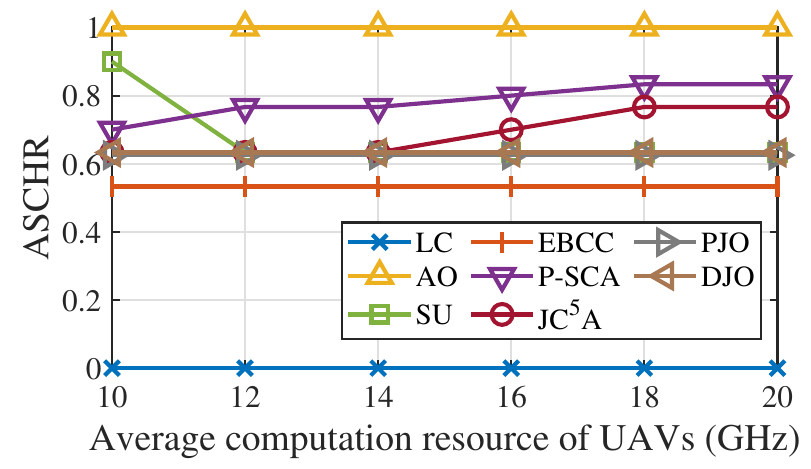}
		\end{minipage}
	}
	\centering
 \vspace{-0.5em}
\caption{System performance with different average computation resources of UAVs.}
         \label{fig_1}
         \vspace{-1.3em}
\end{figure*}

\begin{figure*}[!hbt] 
	\centering
	\subfigure[Average completion delay]
	{
		\begin{minipage}[t]{0.31\linewidth}
			\raggedleft
			\includegraphics[width=2.3in]{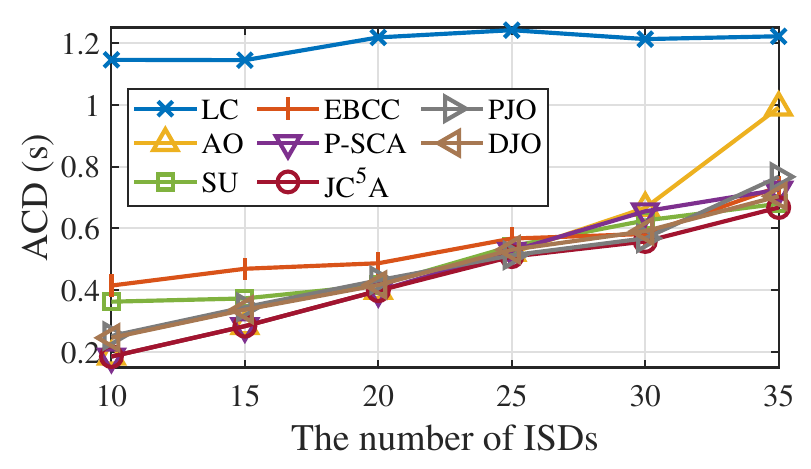}
		\end{minipage}
	}
	\subfigure[Average processing rate]
	{
		\begin{minipage}[t]{0.31\linewidth}
			\centering
			\includegraphics[width=2.3in]{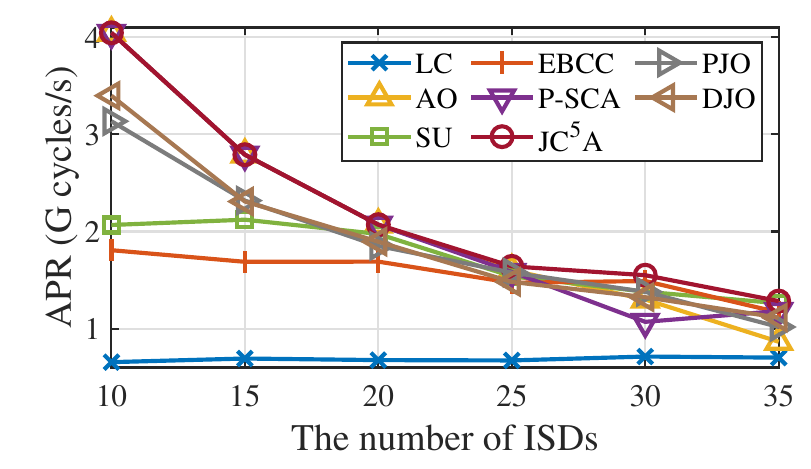}	
		\end{minipage}
	}
	\subfigure[Average service caching hit ratio]
	{
		\begin{minipage}[t]{0.31\linewidth}
			\centering
			\includegraphics[width=2.3in]{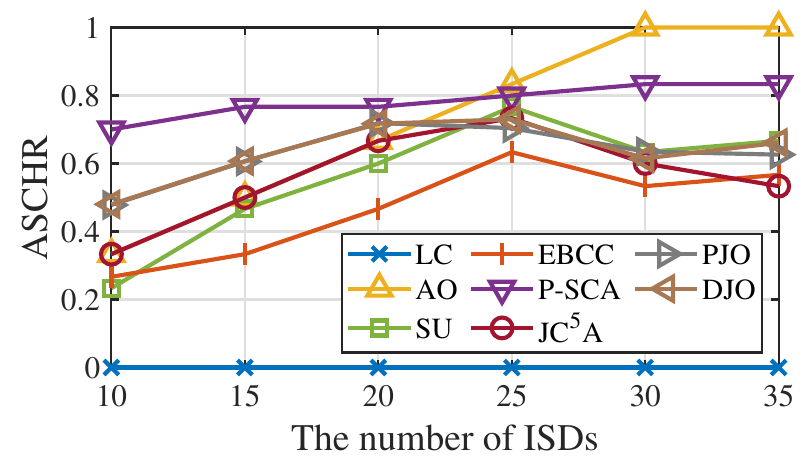}
		\end{minipage}
	}
	\centering
  \vspace{-0.5em}
	\caption{System performance with different numbers of ISDs.}
 	\label{fig_2}

 \vspace{-1.4em}
\end{figure*}

\par \textbf{\textit{Impact of ISD Numbers.}} Figs. \ref{fig_2}(a), \ref{fig_2}(b), and \ref{fig_2}(c) display the impact of the number of ISDs on ACD, APR, and ASCHR, respectively, for the abovementioned eight approaches. It can be observed from Figs. \ref{fig_2}(a) and \ref{fig_2}(b) that the proposed J$\text{C}^5$A outperforms the benchmark approaches in terms of ACD and APR with the increasing of the ISDs. However, it exhibits a medium performance level in ASCHR as the number of ISDs grows, and several factors contribute to these outcomes. First, the inferior performance of LC is because each resource-constrained ISD processes tasks locally. Moreover, the computation offloading of AO, the UAV trajectory control of SU, and the resource allocation of EBCC exhibit inefficiencies due to their static and simplified strategies. Furthermore, although P-SCA achieves relatively superior ASCHR, the centralized optimization of computation offloading and service caching results in increased computation costs. Besides, PJO and DJO may require extensive environmental interactions to achieve the optimal outcomes since the mixed-integer and coupled decision space becomes potentially large as the number of ISDs increases, leading to higher computational complexity and long service delay.

\par Comparatively, J$\text{C}^5$A can achieve the exceptional strategies by managing the mutual-coupled and mix-integer decisions through decoupling the original problem into computationally light weight subproblems. Although J$\text{C}^5$A has a lower ASCHR, it can successfully minimize the service delay by decreasing the frequency of cache replacements. In summary, the proposed J$\text{C}^5$A exhibits superior adaptability to dense scenarios, achieving superior performance in ASD and APR with a trade off in reduced ASCHR.

\par \textbf{\textit{Impact of UAV Service Caching Storage.}} Figs. \ref{fig_3}(a), \ref{fig_3}(b), and \ref{fig_3}(c) illustrate the impact of the average service caching storage of UAVs on ACD, APR, and ASCHR for the eight approaches. First, it can be observed that LC consistently exhibits the worst performance as the service caching storage of UAVs increases, which is obvious due to the independent of the caching capability of UAVs. Moreover, AO, SU, EBCC, P-SCA, PJO, DJO, and J$\text{C}^5$A show an overall downward trend in ACD, an upward trend in APR, and a deceasing tendency in ASCHR. This is because the number of services rises with the increasing of the UAV service caching storage. This is due to the fact that an increase in UAV service caching storage capacity enables the storage of a broader range and a larger number of services. Particularly, AO shows significant variations as the service caching storage of UAVs extends, which can be attributed to the highly dependence of AO on the caching resources of UAVs. In addition, the proposed J$\text{C}^5$A demonstrates superior performance in terms of ACD and APR, while exhibiting relative moderate performance in ASCHR. The main reasons are two fold. On the one hand, the proposed J$\text{C}^5$A aims to minimize the service delay by exploiting the collaborative computing capabilities of MEC servers. On the other hand, J$\text{C}^5$A employs a computationally light weight method for problem solving to ensure the delay for ISDs. In conclusion, as the service caching storage of UAVs increases, despite a decrease in the service caching hit ratio, the proposed J$\text{C}^5$A continues to deliver the superior performances in ASD and APR.

\begin{figure*}[!hbt] 
	\centering
	\subfigure[Average completion delay]
	{
		\begin{minipage}[t]{0.31\linewidth}
			\raggedleft
			\includegraphics[width=2.3in]{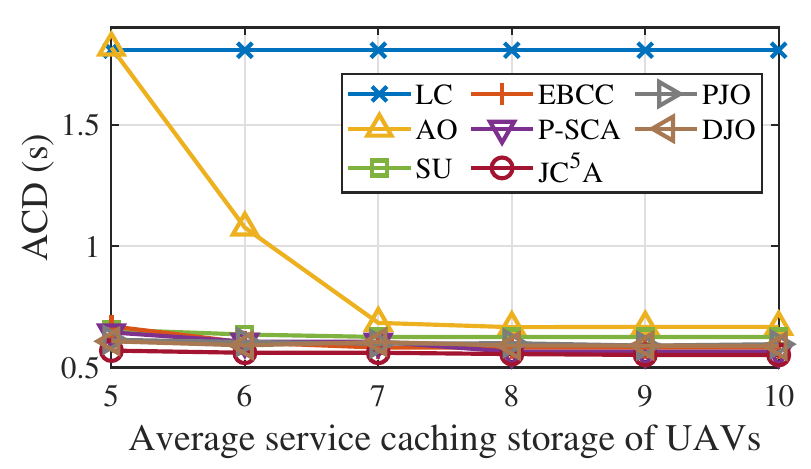}
		\end{minipage}
	}
	\subfigure[Average processing rate]
	{
		\begin{minipage}[t]{0.31\linewidth}
			\centering
			\includegraphics[width=2.3in]{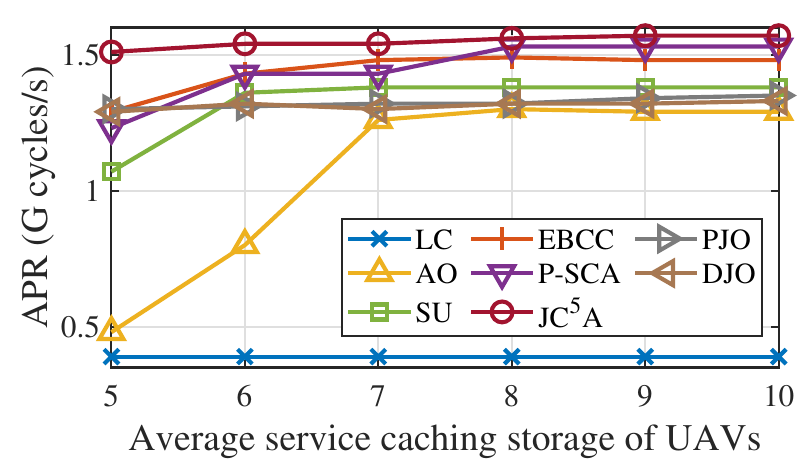}	
		\end{minipage}
	}
	\subfigure[Average service caching hit ratio]
	{
		\begin{minipage}[t]{0.31\linewidth}
			\centering
			\includegraphics[width=2.3in]{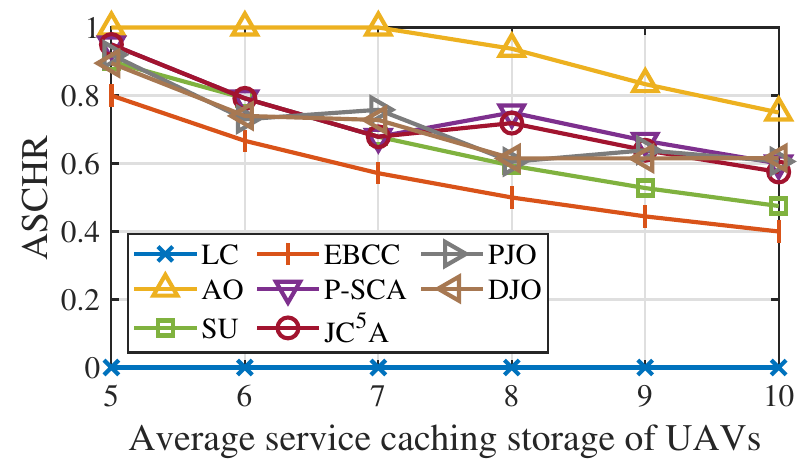}
		\end{minipage}
	}
	\centering
  \vspace{-0.5em}
	\caption{System performance with average service caching storage of UAVs.}
 	\label{fig_3}
	\vspace{-1.5 em}
\end{figure*}

\section{Conclusion}
\label{sec:Conclusion}

\par In this work, we have studied collaborative computation offloading, caching, communication, computation, and trajectory control in an aerial MEC-assisted ICPS. First, we  have designed an aerial-terrestrial UAV collaborative architecture, which consists of an MBS, a cluster of UAVs, and multiple ground ISDs. Moreover, we formulated the SDMOP to minimize the total system delay by jointly optimizing computation offloading, service caching, communication resource allocation, computation resource allocation, and UAV trajectory. Then, we have proposed a J$\text{C}^5$A, which is proved to be converged, to solve the formulated optimization problem. Simulation results have clearly demonstrated that despite sacrificing a portion of ASCHR, the proposed J$\text{C}^5$A exhibits superior performance in ASD and APR to guarantee the computing services for ISDs, and it also shows superior adaptability to dense scenarios of the ICPS.

\bibliographystyle{IEEEtran}
\bibliography{manuscript.bbl}

\end{document}